\documentclass[11pt, letter, notitlepage]{article}
\usepackage{geometry}
\usepackage[utf8]{inputenc} 
\usepackage[T1]{fontenc}    
\usepackage{url}            
\usepackage{booktabs}       
\usepackage{nicefrac}       
\usepackage{microtype}      

\usepackage{latexsym}
\usepackage{color}
\usepackage{xcolor}
\usepackage[numbers]{natbib}
\usepackage{enumerate}
\usepackage{graphicx}
\usepackage{framed}
\PassOptionsToPackage{hyphens}{url}\usepackage{hyperref}
\usepackage{geometry}
\usepackage{amsmath,amsfonts,amssymb,amsthm}
\usepackage{mathtools}
\usepackage{xspace}
\usepackage{float}
\usepackage{comment}
\usepackage{times}
\usepackage{enumitem}
\usepackage{bm}
\usepackage{tikz}
\usepackage{booktabs}
\usepackage[flushleft]{threeparttable}
\usepackage{tabularx}
\usepackage{subfig}
\usepackage{mleftright}
\usepackage{algorithm}
\usepackage{algpseudocode}

\DeclareMathOperator*{\argmax}{arg\,max}
\DeclareMathOperator*{\argmin}{arg\,min}

\newtheorem{theorem}{Theorem}[section]
\newtheorem{lemma}[theorem]{Lemma}
\newtheorem{corollary}[theorem]{Corollary}

 \newtheorem{claim}[theorem]{Claim}

\newtheorem{remark}[theorem]{Remark}
\newtheorem{proposition}[theorem]{Proposition}
\newtheorem{quest}{Question}

\newenvironment{proofof}[1]{\smallskip\noindent{\bf Proof of #1}:}{$\hfill \Box$\\}

\newlength{\boxwidth}
\makeatletter
\DeclareRobustCommand{\qed}{%
  \ifmmode 
  \else \leavevmode\unskip\penalty9999 \hbox{}\nobreak\hfill
  \fi
  \quad\hbox{\qedsymbol}}

\newcommand{\set}[1]{\ensuremath{\left\{ #1 \right\}}}
\newcommand{\norm}[1]{\left\lVert#1\right\rVert}

\newcommand{\hide}[1]{}

\newcommand{\R}{{\ensuremath {\mathbb R}}}
\newcommand{\A}{{\ensuremath {\mathcal A}}}

\newcommand{\PP}{{\ensuremath {\mathcal P}}}
\newcommand{\sX}{\Omega}

\def\tr{\hbox{\rm tr}}
\def\rank{\hbox{\rm rank}}
\def\diam{\hbox{\rm diam}}
\def\sp{\hbox{\rm span}}
\def\fpca{\textsc{Fair-PCA}}
\def\gfpca{{\textsc{Multi-Criteria-Dimension-Reduction}}}
\def\sdpit{\textsc{Iterative-SDP}}

\renewcommand{\S}{{\ensuremath {\mathcal S}}}
\def \Alg2 {\textsc{-SDP}}

\newcommand{\cut}[1]{}

\newcommand{\pr}[1]{\left(#1\right)}

\newcommand{\an}[1]{\left\langle#1\right\rangle}
\newcommand{\floor}[1]{\ensuremath{\left\lfloor#1\right\rfloor}}
\newcommand{\diag}{\textup{diag}}
\def\RR{\mathbb{R}}
\def \OPT  {\mbox{\rm OPT}}

\def \sdp {\mathbb{SDP(I)}}
\def \sdpII {\mathbb{SDP(II)}}
\def \sdpa {\mathbb{SDP}}
\def \sdpr {\mathbb{SDP}(r)}

\title{Multi-Criteria Dimensionality Reduction with \\ Applications to Fairness}

\author{Uthaipon (Tao) Tantipongpipat\thanks{Georgia Institute of Technology} \qquad Samira Samadi\footnotemark[1] \qquad Mohit Singh\footnotemark[1] \\  Jamie Morgenstern\thanks{University of Washington} \qquad Santosh Vempala\footnotemark[1]}

\begin{document}

\maketitle

\begin{abstract}

Dimensionality reduction is a classical technique widely used for data analysis. One foundational instantiation is Principal Component Analysis (PCA), which minimizes the average reconstruction error. In this paper, we introduce the {\em multi-criteria dimensionality reduction} problem where we are given multiple objectives that need to be optimized simultaneously. As an application, our model captures several fairness criteria for dimensionality reduction such as our novel Fair-PCA problem and the Nash Social Welfare (NSW) problem. In  Fair-PCA, the input data is divided into $k$ groups, and the goal is to find a single
$d$-dimensional representation for all groups for which the minimum
variance of any one group is maximized. In NSW, the goal is to maximize the product of the individual variances of the groups achieved by the common low-dimensional space.

Our main result is an exact polynomial-time algorithm for
the two-criterion dimensionality reduction problem when the two criteria are increasing concave functions. As an application of this result, we obtain a polynomial time algorithm for Fair-PCA for $k=2$ groups and a polynomial time algorithm for NSW objective for $k=2$ groups. We also give approximation algorithms for $k>2$.
Our technical contribution in the above results is to prove new low-rank properties of extreme point solutions to semi-definite programs. We conclude with experiments indicating the effectiveness of algorithms based on extreme point solutions of semi-definite programs on several real-world data sets.
\end{abstract}


\section{Introduction}

Dimensionality reduction is the process of choosing a low-dimensional representation of a large, high-dimensional data set. It is a core primitive for modern machine learning and is being used in image processing, biomedical research, time series analysis, etc. Dimensionality reduction can be used during the preprocessing of the data to reduce the computational burden as well as at the final stages of data analysis to facilitate data summarization and data visualization \citep{raychaudhuri1999principal,iezzoni1991applications}. Among the most ubiquitous and effective of dimensionality reduction techniques in practice are Principal Component Analysis (PCA) \citep{kpfrs1901lines,jolliffe1986principal,hotelling1933analysis}, multidimensional scaling \citep{kruskal1964multidimensional}, Isomap \citep{tenenbaum2000global}, locally linear embedding \citep{roweis2000nonlinear}, and t-SNE \citep{maaten2008visualizing}.

One of the major obstacles to dimensionality reduction tasks in
practice is complex high-dimensional data structures that lie on
multiple different low-dimensional subspaces. For example,
\citet{maaten2008visualizing} address this issue for low-dimensional
visualization of images of objects from diverse classes seen from
various viewpoints. Dimensionality reduction algorithms may optimize one data structure well while performs poorly on the others. In this work, we consider when those data structures lying on different low-dimensional subspaces are subpopulations partitioned by sensitive attributes, such as gender, race, and education level.  

As an illustration, consider applying PCA on a high-dimensional data to do a visualization analysis in low dimensions. Standard PCA aims to minimize the single criteria of average reconstruction error over the whole data, but the reconstruction error on different parts of data can be  different. In particular, we  show in Figure~\ref{fig:wpca}  that PCA on the real-world labeled faces in the
wild data set (LFW) \citep{LFWTech} has higher reconstruction error for women than men, and this disparity in performance remains even if male and female faces are sampled with equal weight. We similarly observe  difference in reconstruction errors  of PCA in other real-world datasets. Dissimilarity of performance on different data structure, such as unbalanced average reconstruction errors we demonstrated, raises ethical and legal concerns whether   outcomes of algorithms discriminate the subpopulations against sensitive attributes.
\begin{figure*}[t]
\centering
  \includegraphics[width=0.4\linewidth]{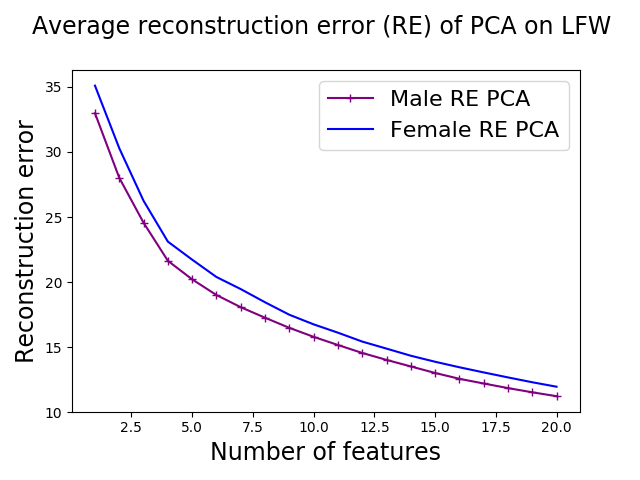}
  \hspace{.5cm}
  \includegraphics[width=0.4\linewidth]{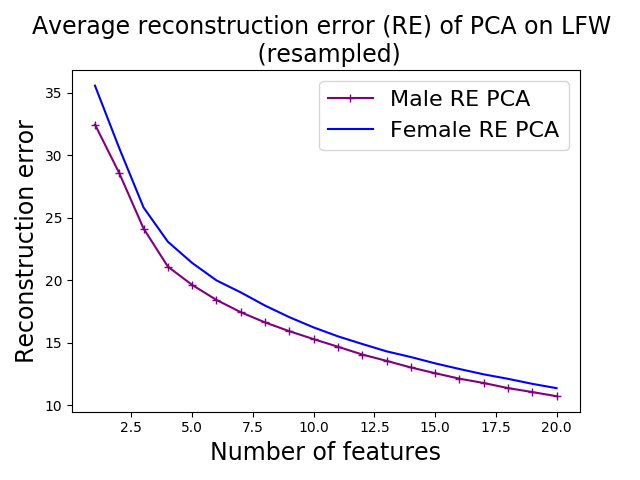}
  \caption{Left: average reconstruction error of PCA on labeled faces in the wild data set (LFW), separated by
  gender. Right: the same, but sampling 1000 faces with men and women equiprobably 
  (mean over 20 samples).}
  \label{fig:wpca}
\end{figure*}

\paragraph{Relationship to fairness in machine learning.}
 In recent years, machine learning community has witnessed an onslaught of
charges  that real-world machine learning algorithms have produced
``biased'' outcomes. The examples come from diverse and impactful
domains.  Google Photos labeled African Americans as
gorillas~\citep{twitter,simonite_2018} and returned queries for CEOs
with images overwhelmingly male and white~\citep{kay2015unequal},
searches for African American names caused the display of arrest
record advertisements with higher frequency than searches for white
names~\citep{Sweeney13}, facial recognition has wildly different
accuracy for white men than dark-skinned women~\citep{gendershades},
and recidivism prediction software has labeled low-risk African
Americans as high-risk at higher rates than low-risk white
people~\citep{propublica}.

The community's work to explain these observations has roughly fallen
into either ``biased data'' or ``biased algorithm'' bins.  In some cases, the
training data might under-represent (or over-represent) some group, or
have noisier labels for one population than another, or use an
imperfect proxy for the prediction label (e.g., using arrest records
in lieu of whether a crime was committed). Separately, issues of
imbalance and bias might occur due to an algorithm's behavior,
such as focusing on accuracy across the entire distribution rather
than guaranteeing similar false positive rates across populations, or
by improperly accounting for confirmation bias and feedback loops in
data collection. If an algorithm fails to distribute loans or bail to
a deserving population, the algorithm won't receive additional data
showing those people would have paid back the loan, but it will
continue to receive more data about the populations it (correctly)
believed should receive loans or bail.

Many of the proposed solutions to ``biased data'' problems amount to
re-weighting the training set or adding noise to some of the labels;
for ``biased algorithms,'' most work has focused on maximizing
accuracy subject to a constraint forbidding (or penalizing) an unfair
model. Both of these concerns and approaches have significant merit,
but form an incomplete picture of the machine learning pipeline  where unfairness
might be introduced therein. Our work takes another step in fleshing
out this picture by analyzing when \emph{dimensionality reduction}
might inadvertently introduce bias.

This work underlines the importance of considering fairness and bias
at every stage of data science, not only in gathering and documenting
a data set~\citep{datasheets} and in training a model, but also in any
interim data processing steps. Many scientific disciplines have
adopted PCA as a default preprocessing step, both to avoid the curse
of dimensionality and also to do exploratory/explanatory data analysis
(projecting the data into a number of dimensions that humans can more
easily visualize). The study of human biology, disease, and the
development of health interventions all face both aforementioned
difficulties, as do numerous economic and financial analysis. In such
high-stakes settings, where statistical tools will help in making
decisions that affect a diverse set of people, we must take particular
care to ensure that we share the benefits of data science with a
diverse community.


We also emphasize this work has implications for representational
rather than just allocative harms, a distinction drawn by 
\citet{kate17} between how people are represented and what
goods or opportunities they receive. Showing primates in search
results for African Americans is repugnant primarily due to its
representing and reaffirming a racist painting of African Americans,
not because it directly reduces any one person's access to a
resource. If the default template for a data set begins with running
PCA, and PCA does a better job representing men than women, or white
people over minorities, the new representation of the data set itself
may rightly be considered an unacceptable sketch of the world it aims
to describe.

\begin{remark}
  We focus on the setting where we ask for a single projection into
  $d$ dimensions rather than  separate projections for each group, because using
   distinct projections (or more generally distinct models) for
  different populations raises legal and ethical concerns.
 \footnote{\citet{lipton2018does}  
    have asked whether equal treatment requires different models for
    two groups.}
  \end{remark}
  
\paragraph{Instability of PCA.}
Disparity of performance in subpopulations of  PCA is closely related to its instability. Maximizing total variance or equivalently minimizing total reconstruction errors is  sensitive to a slight change of data, giving widely different outcomes even if data are sampled from the same distribution.
An example is shown in Figure \ref{fig:unstable-PCA}. Figure \ref{fig:unstable-PCA-2D} shows the distribution of two groups lying in orthogonal dimensions. When the first group's variance in the x-axis is slightly higher than the second group's variance in the y-axis, PCA outputs  the x-axis,  otherwise it outputs the y-axis, and it rarely outputs something in between. The instability of performance can be shown in Figure \ref{fig:fair-var}. Even though in each trial, data are sampled from the same distribution, PCA solutions are unstable and give oscillating variances to each group. However,  solutions to one of our proposed  formulations (\fpca{}, which is to be presented later) are stable and give the same, optimal variance to both groups in each trial. 


Our work presents a novel general
framework that addresses all aforementioned issues: data lying on different low-dimensional structure, unfairness, and instability of PCA. A common difficulty in those settings is that a single
criteria for  dimensionality reduction might not be sufficient to
capture different structures in the data. This motivates our study of
multi-criteria dimensionality reduction. 
 
 \begin{figure}[!tbp]
  \centering
  \subfloat[A distribution of two groups where PCA into one dimension is unfair and unstable]
{\includegraphics[width=0.4\textwidth]{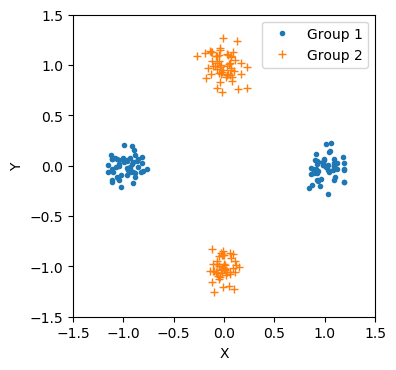}
\label{fig:unstable-PCA-2D}}
  \hfill
  \subfloat[Variances of two groups by PCA are unstable as data are resampled from the same distribution across many trials. Large gaps of two variances results from PCA favoring one group and ignoring the other. Fair-PCA equalizes two variances and is stable over resampling.]
{\includegraphics[width=0.54\textwidth]{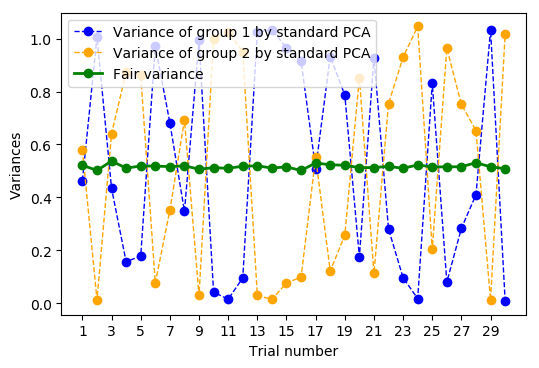}
\label{fig:fair-var}}
  \caption{An example of the distribution of two groups which has very unstable and unfair PCA output}
  \label{fig:unstable-PCA}
\end{figure}

\paragraph{Multi-criteria dimensionality reduction.}
Multi-criteria dimensionality reduction could be used as an umbrella
term with specifications changing based on the applications and the metrics that the machine learning researcher has in mind. Aiming for an output with a balanced error over different subgroups seems to be a natural choice, extending economic game theory literature. 
For example, this covers maximizing geometric mean of the variances of the groups, which is the well-studied Nash social welfare (NSW) objective \citep{kaneko1979nash,nash1950bargaining}. Motivated by these settings, the more general question that we would like to study is as follows.
\begin{quest}
How might one redefine dimensionality reduction to produce projections which optimize different groups' representation in a balanced way?
\end{quest}

For simplicity of explanation, we first describe our framework for
PCA, but the approach is general and applies to a much wider class
of dimensionality reduction techniques. Consider the data points as
rows of an $m \times n$ matrix $A$. For PCA, the objective is to find
an $n \times d$ projection matrix $P$ that maximizes the Frobenius
norm $\|AP\|_F^2$ (this is equivalent to minimizing the reconstruction error \(\|A-APP^T\|_F^2\)). Suppose that the rows of $A$ belong to different \emph{groups} based
on demographics or some other semantically meaningful clustering.  The
definition of these groups need not be a partition; each group could
be defined as a different weighting of the data set (rather than a
subset, which is a 0/1 weighting). Multi-criteria dimensionality
reduction can then be viewed as simultaneously considering objectives
on the different weightings of $A$, i.e., $A_i$.  One way to balance multiple
objectives is to find a projection $P$ that maximizes the minimum
objective value over each of the groups (weightings):
\[
\max_{P\in\R^{n\times d}:P^TP=I_d}\min_{1 \le i \le k} \|A_i P\|_F^2= \langle A_i^T  A_i , PP^T \rangle .  \tag{\fpca}
\]
More generally, let $\PP_d$ denote the set of all $n \times d$
projection matrices $P$, i.e., matrices with $d$ orthonormal
columns. For each group $A_i$, we associate a function
$f_i:\PP_d \rightarrow \R$ that denotes the group's objective value
for a particular projection.
We are also given an accumulation function $g:\R^k \rightarrow \R$. We
define the $(f,g)$-multi-criteria dimensionality reduction problem as finding a
$d$-dimensional projection $P$ which optimizes
\[
\max_{P\in \PP_d} g(f_1(P),f_2(P), \ldots, f_k(P)).  \tag{\gfpca{}}
\]
In the above example of \fpca{}, $g$ is simply the $\min$
function and $f_i(P)=\|A_i P\|^2$ is the total squared norm of the
projection of vectors in $A_i$. The central motivating questions of this paper
are the following:
\begin{itemize}
\item 
\textit{What is the complexity of \fpca?}
\item
\textit{More generally, what is the complexity of \gfpca{}?}
\end{itemize}


Framed another way, we ask whether these multi-criteria optimization
problems force us to incur substantial computational cost compared to
optimizing $g$ over $A$ alone.

\hide{
\citet{fairpcanips18} introduced the
problem of \fpca{} and showed how to use the natural semi-definite
relaxation to find a rank-($d+k-1$) approximation whose cost is at most
that of the optimal rank-$d$ approximation. For $k=2$ groups, this is
an increase of $1$ in the dimension (as opposed to the na\"{i}ve bound
of $2d$, by taking the span of the optimal $d$-dimensional subspaces
for the two groups). The computational complexity of finding the exact
optimal solution to \fpca{} was left as an open question.
}

\paragraph{Summary of contributions.}
We summarize our contributions in this work as follows.
\begin{enumerate}
\item 
We introduce a novel definition of \gfpca{}. 

\item We give polynomial-time algorithms for \gfpca{} with provable guarantees.
\item We analyze the complexity and show hardness of \gfpca{}.
\item We present empirical results to show efficacy of  our algorithms in addressing fairness.
\end{enumerate}
\hide{
\citet{fairpcanips18} introduced the
problem of \fpca{} and showed how to use the natural semi-definite
relaxation to find a rank-($d+k-1$) approximation whose cost is at most
that of the optimal rank-$d$ approximation. For $k=2$ groups, this is
an increase of $1$ in the dimension (as opposed to the na\"{i}ve bound
of $2d$, by taking the span of the optimal $d$-dimensional subspaces
for the two groups). The computational complexity of finding the exact
optimal solution to \fpca{} was left as an open question.
}

We have introduced \gfpca{} earlier, and we now present the technical contributions in this paper.
\subsection{Summary of technical results}

Let us first focus on \fpca{} for ease of exposition. The problem can be
reformulated as the following mathematical program where we denote
$PP^T$ by $X$.  A natural approach to solving this problem is to
consider the SDP relaxation obtained by relaxing the $\rank$
constraint to a bound on the trace.

\begin{center}
\begin{tabular}[h]{|p{6.5cm}|p{6.5cm}|} \hline 
\vspace{0.5em}\qquad\qquad \ \textbf{Exact \fpca{} }{\begin{align*}
  \qquad &\max \, \, z  \label{eq:fpca-relax-top} \\
\langle A_i^TA_i, X\rangle & \ge z \quad i \in \{1,\ldots, k\}\\
\rank(X) & \le d\\
0 \preceq \, X  &\preceq I
\end{align*}} & \vspace{0.5em}\qquad \textbf{SDP relaxation of \fpca{} }{\begin{align}
  \qquad &\max \, \, z \\
\langle A_i^TA_i, X\rangle & \ge z \quad i \in \{1,\ldots, k\}\\
\tr(X) & \le d\\
0 \preceq \, X  &\preceq I  \label{eq:fpca-relax-bottom}
\end{align}} \\ \hline
\end{tabular}
\end{center}

\cut{
\begin{align*}
\max \, & \, z \\
\langle A_i^TA_i, X\rangle & \ge z \quad i \in \{1,\ldots, k\}\\
\rank(X) & \le d\\
0 \preceq \, X  &\preceq I
\end{align*}}

\cut{
\begin{align*}
\max \, & \, z \\
\langle A_i^TA_i, X\rangle & \ge z \quad i \in \{1,\ldots, k\}\\
\tr(X) & \le d\\
0 \preceq \, X  &\preceq I
\end{align*}
}

Our first main result is that the SDP relaxation is exact when there
are \emph{two} groups. Thus finding an extreme point of this SDP gives
an exact algorithm for \fpca{} for two groups. 

\begin{theorem}\label{thm:2groups}
  Any optimal extreme point solution to the SDP relaxation for \fpca{}
  with two groups has rank at most $d$. Therefore, $2$-group \fpca{} can
  be solved in polynomial time.
\end{theorem}

Given $m$ data points partitioned into $k\leq n$ groups in $n$ dimensions, the algorithm runs in \(O(nm+n^{6.5})\) time. $O(mnk)$ is from computing $A_i^T A_i$ and $O(n^{6.5})$ is from solving an SDP over \(n\times n\) PSD matrices \citep{ben2001lectures}. Alternative heuristics and their analyses are discussed in Section \ref{sec:heuristics}.
Our results also hold for the \gfpca{} when $g$ is monotone
nondecreasing in any one coordinate and concave, and each $f_i$ is an
affine function of $PP^T$ (and thus a special case of a quadratic
function in $P$).

\begin{theorem}\label{thm:2groups2}
  There is a polynomial-time algorithm for $2$-group \gfpca{} 
  when $g$ is concave and monotone nondecreasing for at least one of
  its two arguments and each $f_i$ is linear in $PP^T$, i.e.,
  $f_i(P)=\langle B_i, PP^T\rangle$ for some matrix $B_i(A)$.
\end{theorem}

As indicated in the theorem, the core idea is that extreme-point
solutions of the SDP in fact have rank $d$, not just trace equal to
$d$.
For $k > 2$, the SDP need not recover a rank-$d$ solution. In fact,
the SDP may be inexact even for \(k=3\) (see
Section~\ref{sec:gap}). Nonetheless, we show that we can bound the
rank of a solution to the SDP and obtain the following result. We
state it for $\fpca{}$, although the same bound holds for $\gfpca{}$ under
the same assumptions as in Theorem~\ref{thm:2groups2}. Note that this result
generalizes Theorems \ref{thm:2groups} and~\ref{thm:2groups2}.

\begin{theorem}\label{thm:approx}
  For any concave $g$ that is monotone nondecreasing in at least one
  of its arguments, there exists a polynomial time algorithm for \gfpca{}
  with $k$ groups that returns a
  $ d+\floor{\sqrt{2k+\frac{1}{4}}-\frac{3}{2}}$-dimensional embedding
  whose objective value is at least that of the optimal
  $d$-dimensional embedding. If $g$ is only concave, then the solution
  lies in at most $d+1$ dimensions.
\end{theorem}

We  note that the iterative rounding framework for linear programs \cite{lau2011iterative} would give a rank bound of $d+k-1$ for the \fpca{} problem (see \cite{samadi2018price} for details). Hence, we strictly improves  the bound to \(d+\floor{\sqrt{2k+\frac{1}{4}}-\frac{3}{2}}\).  
Moreover, if the dimensionality of the
solution is a hard constraint, instead of tolerating \(s=O(\sqrt{k})\)
extra dimension in the solution, one may solve \fpca{} for target
dimension \(d-s\) to guarantee a solution of rank at most \(d\). Thus,
we obtain an approximation algorithm for \fpca{} of factor
\(1-\frac{O(\sqrt{k})}{d}\).

\begin{corollary} \label{cor:approx-SDP2}
Let  \(A_1,\ldots,A_k\) be data sets of \(k\) groups and suppose \(s:=\floor{ \sqrt{2k+\frac{1}{4}} -\frac32} <d\). Then there exists a polynomial-time approximation algorithm of factor \(1-\frac{s}{d}=1-\frac{O(\sqrt{k})}{d}\) to  \fpca{}.
\end{corollary}

That is, the algorithm returns a projection \(P\in\PP_d\) of \textit{exact} rank \(d\) with objective at least \(1-\frac sd\) of the optimal objective. More details on the approximation result are in Section~\ref{sec:approx-alg}.  The runtime of Theorems~\ref{thm:2groups2} and \ref{thm:approx} depends on the access to first order oracle to \(g\), and standard application of the ellipsoid algorithm would take $\tilde{O}(n^2)$ oracle calls.

 We also develop a general
 rounding framework for SDPs with eigenvalue upper bounds and
$k$ other linear constraints. This algorithm gives a solution of
desired rank that violates each constraint by a bounded amount.  It implies
that for \fpca{} and some of its variants, the additive error
is
\[\Delta(\A):=\max_{S\subseteq [m]} \sum_{i=1}^{\lfloor \sqrt{2|S|}+1\rfloor} \sigma_i(A_S)\]
where $A_S = \frac{1}{|S|} \sum_{i\in S} A_i$. The
precise statement is Theorem \ref{thm:low-rank-iterative} and full details are presented in Section \ref{sec:approx}.

It is natural to ask whether \fpca{} is NP-hard to solve exactly. The
following result implies that it is, even for target dimension $d=1$.
\begin{theorem}\label{thm:NPhard}
  The \fpca{} problem for target dimension $d=1$ is NP-hard
  when the number of groups $k$ is part of the input.
\end{theorem}

This raises the question of the complexity for constant $k \geq 3$
groups. 
%
For $k$ groups, we would have $k$ constraints, one for each group,
plus the eigenvalue constraint and the trace constraint; now the
tractability of the problem is far from clear. In fact, as we show in
Section \ref{sec:gap}, the SDP has an integrality gap even for
$k=3,d=1$. We therefore consider an approach beyond SDPs, to one that
involves solving non-convex problems. Thanks to the powerful
algorithmic theory of quadratic maps, developed
by~\citet{grigoriev2005polynomial}, it is polynomial-time solvable to
check feasibility of a set of quadratic constraints for any fixed
$k$. As we discuss next, their algorithm can check for zeros of a
function of a set of $k$ quadratic functions, and can be used to
optimize the function. Using this result, we show that for $d=k=O(1)$,
there is a polynomial-time algorithm for rather general functions $g$
of the values of individual groups.

\begin{theorem}\label{thm:fixed-kd}
  Let  $g:\R^k \rightarrow \R$ where $g$ is a
  degree-$\ell$ polynomial in some computable subring of $\R^k$, and
let  each $f_i$ be quadratic for $1\leq i\leq k$. Then there is an
  algorithm to solve \((f,g)\)-\gfpca{} in time
  $(\ell d n)^{O(k+d^2)}$.
\end{theorem}

By choosing $g$ to be the product polynomial over the usual
$(\times,+)$ ring or the $\min$ function which is degree $k$ in the
$(\min,+)$ ring, this applies to \fpca{} discussed above
and various other problems.

\subsection{Techniques}

\paragraph{SDP extreme points.}
For $k=2$, the underlying structural property we show is that extreme
point solutions of the SDP have rank exactly $d$.
First,
for $k=d=1$, this is the largest eigenvalue problem, since the maximum
obtained by a matrix of trace equal to $1$ can also be obtained by one
of the extreme points in the convex decomposition of this matrix. This
extends to trace equal to any $d$, i.e., the optimal solution must be
given by the top $d$ eigenvectors of $A^TA$. Second, without the
eigenvalue bound, for any SDP with $k$ constraints, there is an upper
bound on the rank of any extreme point, of $O(\sqrt{k})$, a seminal
result of~\citet{pataki1998rank} (see
also~\citet{barvinok1995problems}). However, we cannot apply this
directly as we have the eigenvalue upper bound constraint. The
complication here is that we have to take into account the constraint
$X \preceq I$ without increasing the rank.

\begin{theorem} \label{thm:low-rank}
Let $C$ and $A_1,\ldots, A_m$ be $n \times n$ real matrices, \(d\leq n\), and \(b_1,\ldots b_m\in\RR\). Suppose the semi-definite program \(\sdp\):
\begin{eqnarray}
  \min \langle C, X\rangle \text{ subject to }  \\
   \langle A_i , X\rangle  & \lhd_i & b_i \;\; \forall \; 1\leq i\leq m \label{eq:sdp-con}\\
   \tr(X)&\leq &d\\
   0\preceq X&\preceq& I_n
\end{eqnarray}
 where \({\lhd}_i\in\set{\leq, \geq, =}\), has a nonempty feasible set. Then, all extreme optimal solutions $X^*$ to $\sdp$  have rank at most $r^*:=d +\floor{ \sqrt{2m+\frac{9}{4}} -\frac32}$. Moreover, given a feasible optimal solution, an extreme optimal solution can be found in polynomial time.
\end{theorem}

To prove the theorem, we extend~\citet{pataki1998rank}'s
characterization of rank of SDP extreme points with minimal loss in the
rank. We show that the constraints $0\preceq X \preceq I$ can be
interpreted as a generalization of restricting variables to lie
between $0$ and $1$ in the case of linear programming
relaxations. From a technical perspective, our results give new
insights into structural properties of extreme points of semi-definite
programs and more general convex programs. 
Since the result of \cite{pataki1998rank} has been studied from perspective of fast algorithms~\cite{boumal2016non,burer2003nonlinear,burer2005local} and applied in community detection and phase synchronization~\cite{bandeira2016low}, we expect our extension of the result to have further applications in many of these areas.

\paragraph{SDP iterative rounding.}
Using Theorem~\ref{thm:low-rank}, we extend the iterative rounding framework for linear programs (see ~\cite{lau2011iterative} and references therein) to semi-definite programs, where the $0,1$ constraints are generalized to eigenvalue bounds. The algorithm has a remarkably similar flavor. In each iteration, we fix the subspaces spanned by eigenvectors with $0$ and $1$ eigenvalues, and argue that one of the constraints can be dropped while bounding the total violation in the constraint over the course of the algorithm. While this applies directly to the \fpca{} problem, it is in fact a general statement for SDPs, which we give below.

Let $\A=\{A_1,\ldots, A_m\}$ be a collection of $n\times n$ matrices. For any set $S\subseteq \{1,\ldots, m\}$, let $\sigma_i(S)$ the $i^{th}$ largest singular of the average of matrices $\frac{1}{|S|} \sum_{i\in S} A_i$.
We let $$\Delta(\A):=\max_{S\subseteq [m]} \sum_{i=1}^{\lfloor \sqrt{2|S|}+1\rfloor} \sigma_i(S).$$
\begin{theorem} \label{thm:low-rank-iterative}
Let $C$ be a real $n\times n$ matrix and $\A=\{A_1,\ldots, A_m\}$ be a collection of~real $n \times n$  matrices, \(d\leq n\), and \(b_1,\ldots b_m\in\RR\). Suppose the semi-definite program \(\sdpa\):
\begin{eqnarray*}
  \min \langle C, X\rangle \text{ subject to }  \\
   \langle A_i , X\rangle  & \geq  & b_i \;\; \forall \; 1\leq i\leq m \\
   \tr(X)&\leq &d\\
   0\preceq X&\preceq& I_n
\end{eqnarray*}
  has a nonempty feasible set and let $X^*$ denote an optimal solution. The algorithm \sdpit{} (see Algorithm \ref{alg:sdp} in Section \ref{sec:approx}) returns a matrix $\tilde{X}$ such that \begin{enumerate}
 \item rank of $\tilde{X}$ is at most $d$,
 \item  $\langle C, \tilde{X}\rangle \leq \langle C, X^* \rangle $, and
  \item $\langle A_i , \tilde{X}\rangle   \geq   b_i- \Delta(\A)$ for each $1\leq i\leq m$.
 \end{enumerate}
Moreover, \sdpit{} runs in polynomial time. \end{theorem}
The time complexity of Theorems \ref{thm:low-rank} and \ref{thm:low-rank-iterative} is analyzed in Sections \ref{sec:sdp-round} and \ref{sec:approx}, respectively. Both algorithms introduce the rounding procedures that do not contribute significant computational cost; rather, solving the SDP is the bottleneck for running time both in theory and practice.

\subsection{Organization}
We present related work in Section \ref{sec:related-work}. In Section \ref{sec:sdp-round}, we prove Theorem~\ref{thm:low-rank} and apply the result to \gfpca{} to obtain Theorem \ref{thm:approx}. In Section \ref{sec:fairness-criteria}, we present and motivate several fairness criteria for dimensionality reduction, including a novel one of our own, and apply Theorem \ref{thm:approx} to get approximation algorithm of \fpca{}, thus proving Corollary \ref{cor:approx-SDP2}.
In Section \ref{sec:approx}, we  
 give an iterative rounding algorithm and prove Theorem \ref{thm:low-rank-iterative}. In Section \ref{sec:poly-time}, we show the polynomial-time solvability of \gfpca{} when the number of groups \(k\) and the target dimension \(d\) are fixed, proving Theorem~\ref{thm:fixed-kd}. In Section \ref{sec:hardness-and-gap}, we show NP-hardness and integrality gap of \gfpca{} for \(k>2\).
In Section \ref{sec:experiments}, we show the experimental results of our algorithms on real-world data sets, evaluated by different fairness criteria, and present additional algorithms with improved runtime. We present missing proofs in Appendix \ref{app:proof}. In Appendix \ref{sec:tight}, we  show that the rank of extreme solutions of SDPs in Theorem \ref{thm:low-rank} cannot be improved.
\subsection{Related work} \label{sec:related-work}

\paragraph{Optimization.}
As mentioned earlier,~\citet{pataki1998rank} (see also~\citet{barvinok1995problems}) showed that low rank solutions to semi-definite programs with small number of affine constraints can be obtained efficiently.
 Restricting a feasible region of certain SDPs relaxations with low-rank constraints has been shown to avoid  spurious local optima \citep{bandeira2016low} and reduce the runtime due to known heuristics and analysis  \citep{burer2003nonlinear,burer2005local,boumal2016non}.
We also remark that methods based on Johnson-Lindenstrauss lemma can also be applied to obtain bi-criteria results for the \fpca{} problem. For example,~\citet{so2008unified} give algorithms that give low rank solutions for SDPs with affine constraints without the upper bound on eigenvalues.
Here we have focused on the single criteria setting, with violation either in the number of dimensions or the objective but not both. We also remark that extreme point solutions to linear programming have played an important role in design of approximation algorithms~\citep{lau2011iterative}, and our result add to the comparatively small, but growing, number of applications for utilizing extreme points of semi-definite programs.  

A closely related area, especially to \gfpca{}, is multi-objective
optimization which has a vast literature. We refer the reader
to~\citet{deb2014multi} and references therein. We  remark that
properties of extreme point solutions of linear
programs~\citep{ravi1996constrained,grandoni2014new} have also been
utilized to obtain approximation algorithms to multi-objective
problems. For semi-definite programming based methods, the closest
works are on simultaneous
max-cut~\citep{bhangale2015simultaneous,bhangale2018near} that utilize
sum of squares hierarchy to obtain improved approximation algorithms.

\paragraph{Fairness in machine learning.}The applications of multi-criteria dimensionality reduction in fairness are closely related to studies on representational bias in machine learning \citep{kate17,noble2018algorithms,bolukbasi2016man}, for which there have been various mathematical formulations studied \citep{chierichetti2017fair,celis2018fair,kleindessner2019fair,kleindessner2019guarantees}.
One interpretation of
our work is that we suggest using multi-criteria dimensionality reduction rather than standard PCA when
creating a lower-dimensional representation of a data set for further
analysis. Two most relevant pieces of work take the posture of explicitly trying to reduce the correlation
between a sensitive attribute 
(such as race or gender)
and the new
representation of the data.  The first piece is a broad line of
work ~\citep{representation,Beutel2017DataDA,
  calmon2017optimized,madras18a,zhang2018mitigating} that aims to design representations
which will be conditionally independent of the protected attribute,
while retaining as much information as possible (and particularly
task-relevant information for some fixed classification
task). The second piece is the work by \citet{convexpca}, who also look to design PCA-like maps which
reduce the projected data's dependence on a sensitive attribute. Our
work has a qualitatively different goal: we aim not to hide a
sensitive attribute, but  to instead maintain as much information about
each population after projecting the data.  In other words, we look
for representation with similar richness for population,
rather than making each group indistinguishable.

Other work has developed techniques to obfuscate a sensitive attribute
directly~\citep{PedreshiRT08,KamiranCP10,CaldersV10,KamiranC11,LuongRT11,KamiranKZ12,KamishimaAAS12,HajianD13,FeldmanFMSV15,ZafarVGG15,FishKL16,AdlerFFRSSV16}.
 This line of work diverges
from ours in two ways. First, these works focus on representations
which obfuscate the sensitive attribute rather than a representation
with high fidelity regardless of the sensitive attribute. Second, most of
these works do not give formal guarantees on how much an objective
will degrade after their transformations. Our work gives theoretical guarantees including an exact optimality for two groups.

Much of  other work on fairness for learning algorithms focuses on
fairness in classification or scoring~\citep{fta, hardt2016,
  kleinberg2016, chouldechova2017fair}, or in online learning
settings~\citep{JKMR16, KKMPRVW17, EFNSV17}. These
works focus on either statistical parity of the decision rule, or
equality of false positives or negatives, or an algorithm with a fair
decision rule. All of these notions are driven by a single learning
task rather than a generic transformation of a data set, while our work
focuses on a ubiquitous, task-agnostic preprocessing step.

\paragraph{Game theory applications.}
The applications of multi-criteria dimensionality reduction in fairness are closely related to studies on fair resource allocation in game theory \citep{wei2010game,fang2004fair}.   From the game theory literature, our model covers Nash social welfare objective \citep{kaneko1979nash,nash1950bargaining} and others \citep{kalai1975other,kalai1977proportional}.

\section{Low-rank solutions of \gfpca{}} \label{sec:sdp-round}
In this section, we show that all extreme solutions of SDP relaxation of \gfpca{} have low rank, proving Theorem \ref{thm:2groups}-\ref{thm:approx}. Before we state the results, we make the following assumptions.
In this section, we let \(g:\R^k\rightarrow \R\) be a concave function, and mildly assume that \(g\) can be accessed with a polynomial-time subgradient oracle.
We are explicitly given  functions \(f_1,f_2,\ldots,f_k\) which are affine in \(PP^T\), i.e., we are given real \(n\times n\) matrices \(B_1,\ldots,B_k\) and constants \(\alpha_1,\alpha_2,\ldots,\alpha_k\in\R \), and  \(f_i(P)=\an{B_i,PP^T}+\alpha_i\).

We assume \(g\) to be \(G\)-Lipschitz. For functions \(f_1,\ldots,f_k,g\) that are \(L_1,\ldots,L_k,G\)-Lipschitz, we define an \textit{\(\epsilon\)-optimal} solution to \((f,g)\)-\gfpca{}  as a  matrix \(X\in\R^{n \times n}\) of rank \(d\) with \(0\preceq X \preceq I_n\) whose objective value is at most \(G\epsilon\pr{\sum_{i=1}^kL_i^2}^{1/2}\) away from the optimum. IWhen an optimization problem   has affine constraints \(F_i(X)\leq b_i\) where \(F_i\) is \(L_i\)-Lipschitz for all \(i\in\set{1,\ldots,m}\), we also define an \textit{\(\epsilon\)-feasible} solution as a projection matrix \(X\in\R^{n \times n}\) of rank \(d\) with \(
0\preceq X \preceq I_n\) that violates the \(i\)th affine constraint \(F_i(X)\leq b_i\)  by at most \(\epsilon L_i\) for all \(i\).
Note that the feasible region of the problem is implicitly bounded by the constraint \(X\preceq I_n\).

In this work, an algorithm may involve solving an optimization under a matrix linear inequality, whose exact optimal solutions may not be
representable in finite bits of computation. However, we give algorithms that return an \(\epsilon\)-feasible solution whose running time depends polynomially on  \(\log\frac{1}{\epsilon}\) for any \(\epsilon>0\). This is standard for  computational tractability in convex optimization (see, for example, in \cite{ben2001lectures}).
 Therefore, for ease of exposition, we omit the computational error dependent on this \(\epsilon\) to obtain an \(\epsilon\)-feasible and \(\epsilon\)-optimal solution, and define polynomial  time as polynomial in \(n,k\) and \(\log\frac{1}{\epsilon}\).

We first prove Theorem~\ref{thm:low-rank} below. To prove Theorem \ref{thm:2groups}-\ref{thm:approx},  we first  show that extreme point solutions in a general class of semi-definite cone under affine constraints
and \(X\preceq I\) have low rank. The statement builds on a result of \cite{pataki1998rank}, and also generalizes to  SDPs under a constraint \(X\preceq C\) for any given PSD matrix \(C\in\R^{n\times n}\) by the transformation \(X\mapsto C^{\frac12}YC^\frac12\) of the SDP feasible region.
We then apply our result to \gfpca{}, which generalizes   \fpca, and prove Theorem~\ref{thm:approx}, which implies Theorem  \ref{thm:2groups} and \ref{thm:2groups2}.

\begin{proofof}{Theorem~\ref{thm:low-rank}}
Let $X^*$ be an extreme point optimal solution to $\sdp$. Suppose rank of $X^*$, say $r$, is more than $r^*$. Then we show a contradiction to the fact that $X^*$ is extreme. Let $0\leq l\leq r$ of the eigenvalues of $X^*$ be equal to one. If $l\geq  d$, then we have $l=r=d$ since $\tr(X)\leq d$, and we are done. Thus we assume that $l\leq d-1$. In that case, there exist matrices $Q_1\in \RR^{n\times r-l}$, $Q_2\in \RR^{n\times l}$  and a symmetric matrix $\Lambda\in \RR^{(r-l)\times (r-l)}$ such that
$$X^*= \begin{pmatrix}
Q_1  & Q_2
\end{pmatrix}
\begin{pmatrix}
\Lambda & 0\\
0 & I_{l}
\end{pmatrix}
\begin{pmatrix}
Q_1  & Q_2
\end{pmatrix}^{\top} = Q_1 \Lambda Q_1^{\top}  + Q_2 Q_2^T
$$
where $0\prec \Lambda \prec I_{r-l}$, $Q_1^T Q_1=I_{r-l}$, $Q_2^TQ_2=I_l$, and that the columns of $Q_1$ and $Q_2$ are orthogonal, i.e. \cut{$Q_1^T Q_2=0$. In particular} $Q=\begin{pmatrix}
Q_1  & Q_2
\end{pmatrix}$ has orthonormal columns. Now, we have
$$\langle A_i ,X^*\rangle = \langle A_i, Q_1 \Lambda Q_1^{\top} + Q_2 Q_2^{\top}\rangle = \langle Q_1^{\top} A_i Q_1 , \Lambda\rangle + \langle A_i, Q_2Q_2^{\top}\rangle$$
and $\tr(X^*)=  \langle Q_1^{\top}  Q_1 , \Lambda\rangle + \tr(Q_2Q_2^{\top})$
so that \(\langle A_i ,X^*\rangle \) and \(\tr(X^*)\) are linear in \(\Lambda\).

Observe that the set of $s\times s$ symmetric matrices forms a vector space of dimension $\frac{s(s+1)}{2}$ with the above inner product where we consider the matrices as long vectors.  If  $m +1 < \frac{(r-l) (r-l+1)}{2}$, then there exists a $(r-l)\times (r-l)$-symmetric matrix $\Delta\neq 0$ such that $\langle Q_1^{\top} A_i Q_1 , \Delta\rangle=0$ for each $1\leq i\leq m$ and $ \langle Q_1^{\top}  Q_1 , \Delta \rangle
=0$.

But then we claim that $\tilde X = Q_1 (\Lambda \pm \delta \Delta)Q_1^{\top}  + Q_2 Q_2^T$ is feasible for some small $\delta>0$, which implies a contradiction to \(X^*\) being extreme. Indeed, \(\tilde X\) satisfies all the linear constraints by the construction of $\Delta$. Thus it remains to check the eigenvalues of \(\tilde X\).  Observe that
$$ Q_1 (\Lambda \pm \delta \Delta)Q_1^{\top}  + Q_2 Q_2^T=Q \begin{pmatrix}
\Lambda\pm \delta \Delta & 0\\
0 & I_{l}
\end{pmatrix}
Q^\top $$
with orthonormal \(Q\). Thus it is enough to consider the eigenvalues of  $\begin{pmatrix}
\Lambda\pm \delta \Delta & 0\\
0 & I_{l}
\end{pmatrix}.$

Observe that eigenvalues of the above matrix are exactly $l$ ones and eigenvalues of $\Lambda \pm \delta \Delta$. Since eigenvalues of $\Lambda$ are bounded away from $0$ and $1$, one can find a small $\delta>0$ such that the eigenvalues of $\Lambda\pm \delta \Delta$ are bounded away from $0$ and $1$ as well, so we are done.
Therefore, we must have \(m+1 \geq \frac{(r-l) (r-l+1)}{2}\) which implies \(r-l \leq -\frac{1}{2}+\sqrt{2m+\frac{9}{4}}\). By \(l\leq d-1\), we have \(r\leq r^*\).

To obtain the algorithmic result, given feasible \(\bar{X}\), we iteratively reduce \(r-l\) by at least one until  \(m+1 \geq \frac{(r-l) (r-l+1)}{2}\). While \(m +1 < \frac{(r-l) (r-l+1)}{2}\), we obtain $\Delta$ by  Gaussian elimination. Now we want to find the correct value of \(\pm\delta\) so that $\Lambda'=\Lambda\pm\delta\Delta$  takes one of the eigenvalues to zero or one. First, determine the sign of \(\an{C,\Delta}\) to find the correct sign to move \(\Lambda\) that keeps the objective non-increasing, say it is in the positive direction.
Since the feasible set of $\sdp$ is convex and bounded, the ray \(f(t)=Q_1 (\Lambda+t\Delta) Q_1^{\top}  + Q_2 Q_2^\top,t\geq0\) intersects the boundary of feasible region at a unique \(t'>0\). Perform binary search for  \(t'\) up to a desired accuracy, and set \(\delta=t'\). Since $\langle Q_1^{\top} A_i Q_1 , \Delta\rangle=0$ for each $1\leq i\leq m$ and $ \langle Q_1^{\top}  Q_1 , \Delta \rangle
=0$, the additional tight constraint from moving \(\Lambda' \leftarrow \Lambda+\delta\Delta\) to the boundary of the feasible region must be an eigenvalue constraint \(0\preceq X \preceq\ I_n\), i.e., at least one additional eigenvalue is now at 0 or 1, as desired. We apply eigenvalue decomposition to \(\Lambda'\) and update \(Q_1\) accordingly, and repeat.

The algorithm involves at most \(n\) rounds of reducing \(r-l\), each of which involves Gaussian elimination and several iterations  of checking \(0\preceq f(t) \preceq I_n\) (from binary search) which can be done by eigenvalue value decomposition. Gaussian elimination and eigenvalue decomposition can be done in \(O(n^3)\) time, and therefore the total runtime of SDP rounding is \(\tilde O(n^4)\) which is polynomial.
\end{proofof}

One can initially reduce the rank of given feasible \(\bar{X}\) using an LP rounding in \(O(n^{3.5})\) time \cite{samadi2018price} before our SDP rounding. This reduces the number of iterative rounding steps; particularly,  \(r-l\) is further bounded by \(k-1\). The runtime complexity is then \(O(n^{3.5})+\tilde O(kn^3)\).

 The next corollary is another useful fact of the low-rank property and is used  in the analysis of iterative rounding algorithm in Section \ref{sec:approx}. The corollary can be obtained from the bound \(r-l \leq -\frac{1}{2}+\sqrt{2m+\frac{9}{4}}\) in the proof of Theorem~\ref{thm:low-rank}.
\begin{corollary}\label{cor:fractional}
The number of fractional eigenvalues in any extreme point solution $X$ to $\sdp$ is bounded by $ \sqrt{2m+\frac{9}{4}}-\frac12\leq \lfloor \sqrt{2m}  \rfloor + 1$.

\end{corollary}

We are now ready to prove the main result that we can find a low-rank solution for \gfpca{}.

\hide{
Recall that $\PP_d$ is the set of all $n \times d$ projection matrices $P$, i.e., matrices with $d$ orthonormal columns and the \((f,g)\)-\gfpca{}  is to solve
\begin{equation}
\max_{P\in \PP_d} g(f_1(P),f_2(P), \ldots, f_k(P)) \label{eq:gfpca-pf}
\end{equation}
\begin{theorem} \label{cor:convex-lower-rank}
There exists a polynomial-time algorithm to solve \((f,g)\)-\gfpca{} that returns a solution $\hat{X}$ of rank at most $r^*:=d +\floor{ \sqrt{2k+\frac{1}{4}} -\frac32}$ whose objective value is at least that of the optimal $d$-dimensional embedding.
\end{theorem}
}

\begin{proofof}{Theorem~\ref{thm:approx}}
Let \(r^*:=d +\floor{ \sqrt{2k+\frac{1}{4}} -\frac32}\). Given assumptions on \(g\), we write a relaxation of \gfpca{} as follows:
\begin{align}
  \max_{X\in\R^{n \times n}} &g(\an{B_1,X}+\alpha_1, \ldots, \an{B_k,X}+\alpha_k) \text{ subject to } \label{eq:sdp-gen-obj-lin}\\
   \tr(X)&\leq d\\
   0\preceq X&\preceq I_n \label{eq:sdp-gen-con-lin}
\end{align}
Since \(g(x)\) is concave in \(x\in\R^k\) and \(\an{B_i,X}+\alpha_i\) is affine in \(X\in\R^{n \times n}\), we have that \(g\) as a function of \(X\) is also concave in \(X\). By concavity of \(g\) and that the feasible set is convex and bounded, we can solve the convex program  \eqref{eq:sdp-gen-obj-lin}-\eqref{eq:sdp-gen-con-lin} in polynomial time,  e.g. by ellipsoid method, to obtain a (possibly high-rank) optimal solution \(\bar{X}\in\R^{n \times n}\). (In the case that \(g\) is linear, the relaxation is also an SDP and may be solved faster in theory and practice). 

We first assume that \(g\) is monotonic in at least one coordinate, so without loss of generality, we let \(g\) be nondecreasing in the first coordinate. To reduce the rank of \(\bar{X}\), we consider an $\sdpII$:
\begin{eqnarray}
  \max_{X\in\R^{n \times n}} && \an{B_1,X}\text{ subject to } \label{eq:very-low-rank-obj} \\
   \an{B_i,X}   &= & \an{B_i,\bar{X}}  \;\; \qquad \forall \; 2\leq i\leq k \label{eq:sdp-con-2} \\
   \tr(X)& \leq & d\\
   0\preceq X&\preceq& I_n \label{eq:very-low-rank-con}
\end{eqnarray}
 $\sdpII$ has a feasible solution \(\bar{X}\) of objective \(\an{B_1,X}\), and note that there are \(k-1\) constraints in \eqref{eq:sdp-con-2}.
Hence, we can apply the algorithm in Theorem~\ref{thm:low-rank} with \(m=k-1 \) to find an extreme  solution \(X^*\)  of $\sdpII$    of rank at most \(r^*\). Since \(g\) is nondecreasing in \(\an{B_1,X}\), an optimal solution to $\sdpII$ gives objective value at least the optimum of the relaxation and hence at least the optimum of the original \gfpca{}.

If the assumption that \(g\) is monotonic in at least one coordinate is dropped,
 the argument holds  by indexing constraints \eqref{eq:sdp-con-2} in $\sdpII$ for all \(k\) groups instead of \(k-1\) groups.
\end{proofof}

Another way to state Theorem~\ref{thm:approx} is that the number of groups must reach \(\frac{(s+1)(s+2)}{2}\) before additional \(s\) dimensions in the solution matrix \(P\) is required to achieve the optimal objective value.
For \(k=2\),  no additional dimension in the solution is necessary to attain the optimum, which proves Theorem \ref{thm:2groups2}. In particular, it 
applies to \fpca{} with two groups, proving Theorem~\ref{thm:2groups}. We note that the rank bound in Theorem \ref{thm:low-rank} (and thus also the bound in Corollary \ref{cor:fractional}) is tight. An example of the problem instance follows from \cite{bohnenblust1948joint} with a slight modification. We refer the reader to Appendix \ref{sec:tight} for details.

\section{Applications of low-rank solutions in fairness} \label{sec:fairness-criteria}
In this section, we show applications of low-rank solutions of Theorem~\ref{thm:approx} in fairness applications of dimensionality reduction. We  describe existing fairness criteria and motivate our new fairness objective, summarized in Table \ref{tab:fair-criteria}. The new fairness objective appropriately addresses fairness when subgroups have different optimal variances in a low-dimensional space. We note that all  fairness criteria  in this section satisfy the assumption in Theorem \ref{thm:approx}
  that   $g$  is concave and monotone nondecreasing in at least one
  (in fact, all) of its arguments, and thus these fairness objectives can be solved with Theorem~\ref{thm:approx}. We also give approximation algorithm for \fpca{}, proving Corollary \ref{cor:approx-SDP2}. 
\hide{Recall that
given each group's objective   \(f_i:\PP_d\rightarrow\R\)
and an
accumulation  function $g$, \gfpca{} is the problem of
optimizing 
\[ \max_{P\in \PP_d} g(f_1(P),f_2(P), \ldots, f_k(P)).  \tag{\gfpca{}} \]
where $\PP_d$ denote the set of all $n \times d$
matrices with $d$ orthonormal
columns. We mentioned that an example of \gfpca{} is max-min \fpca{} where $g$ is the $\min$
function and $f_i(P)=\|A_i P\|^2$ for a given data matrix \(A_i\). In this section, we discuss other fairness criteria }

\subsection{Application to \fpca{} } \label{sec:approx-alg}
We prove Corollary \ref{cor:approx-SDP2} below. Recall that, by Theorem~\ref{thm:approx}, \(s:=\floor{ \sqrt{2k+\frac{1}{4}} -\frac32}\) additional dimensions for the projection are required  to achieve the optimal objective. One way to ensure that the algorithm outputs \(d\)-dimensional projection is to solve the problem in the lower target dimension \(d-s\), and then apply the rounding algorithm described in Section \ref{sec:sdp-round}.  

\hide{Recall that given \(A_1,\ldots,A_k\), \fpca{} problem is to solve \[\max_{P:P^TP=I_d}\min_{1 \le i \le k} \|A_i P\|_F^2= \langle A_i^T  A_i , PP^T \rangle\]
We state the approximation guarantee and the algorithm formally as follows.
\begin{corollary} \label{cor:approx-SDP}
Let  \(A_1,\ldots,A_k\) be data sets of \(k\) groups and suppose \(s:=\floor{ \sqrt{2k+\frac{1}{4}} -\frac32} <d\). Then there exists a polynomial-time approximation algorithm of factor \(1-\frac{s}{d}=1-\frac{O(\sqrt{k})}{d}\) to  \fpca{} problem.\end{corollary}
}
\begin{proof}[Proof of Corollary \ref{cor:approx-SDP2}]
We find an extreme solution \(X^*\) of the  \fpca{} problem of finding a projection from \(n\) to \(d-s\) target dimensions. By Theorem~\ref{thm:approx}, the rank of \(X^*\) is at most \(d\).

Denote $\OPT_d,X_d^*$ the optimal value and an optimal solution to exact \fpca{}  with target dimension \(d\), respectively. 
 Note that \(\frac{d-s}{d}X_d^*\) is a feasible solution to the \fpca{} relaxation \eqref{eq:fpca-relax-top}-\eqref{eq:fpca-relax-bottom} on target dimension \(d-s\) whose objective is at least \(\frac{d-s}{d}\OPT_{d}\), since the \fpca{} relaxation objective scales linearly with \(X\). Therefore, the optimal value of \fpca{} relaxation of target dimension \(d-s\), which is achieved by \(X^*\) (Theorem \ref{thm:approx}), is  at least \(\frac{d-s}{d}\OPT_{d}\). Hence, we obtain the \((1-\frac{s}{d})\)-approximation.
\end{proof}

\subsection{Welfare economic and NSW}

\begin{table}
\centering
\caption{Examples of fairness criteria to which our results in this work apply. We are given \(A_i\) as the data matrix of group \(i\) in \(n\) dimensions for \(i=1,\ldots,k\) and a target dimension \(d<n\). We denote \(\PP_d=\set{P\in\R^{n\times d}:P^TP=I_d}\) the set of all $n \times d$
matrices with $d$ orthonormal
columns and \(\beta_i=\max_{Q\in\PP_d} \|A_i Q\|^2 \) the  variance of an optimal projection for group \(i\) alone. }
\label{tab:fair-criteria}
\begin{tabular}{|c|c|c|c|}\hline
Name & \(f_i(P)\) & \(g\) & \gfpca{}  \\\hline
Standard PCA & \(\|A_i P\|^2\) & sum & \(\max_{P\in\PP_d} \sum_{i\in[k]}\|A_i P\|^2\)\\\hline
\fpca{} (MM-Var)  &  \(\|A_i P\|^2\) & min & \(\max_{P\in\PP_d} \min_{i\in[k]}\|A_i P\|^2\) \\\hline
Nash social welfare (NSW) &  \(\|A_i P\|^2\) & product & \(\max_{P\in\PP_d} \prod_{i\in[k]}\|A_i P\|^2\)  \\\hline
Marginal loss (MM-Loss) & \(\|A_i P\|^2 -\beta_i\) & min & \(\max_{P\in\PP_d} \min_{i\in[k]}\pr{\|A_i P\|^2-\beta_i}\)\\\hline
\end{tabular}

\end{table}

If we interpret \(f_i(P)=\|A_i P\|^2\) in \fpca{} as individual utility, then  standard PCA maximizes the total utility of individuals, also known as a utilitarian objective in welfare economic.   One other objective is egalitarian, aiming to maximize the minimum utility  
\cite{kalai1977proportional}, which is equivalent to  \fpca{} in our setting.  
One other example, which lies between the two, is to choose
the product function $g(y_1, \ldots, y_k) = \prod_i y_i$ for the accumulation function  $g$.  This is also a natural choice, famously
introduced in Nash's solution to the bargaining problem \citep{nash1950bargaining,kaneko1979nash}, and we call this objective Nash Social Welfare (NSW). The three objectives are special cases of the $p$th power mean of individual utilities,
i.e. $g(y_1, \ldots, y_k) = \left(\sum_{i \in [k]} y^{p}_i\right)^{1/p}$, with \(p=1,-\infty,0\)  giving  standard PCA, \fpca{}, and NSW, respectively. Since the  \(p\)th power mean is concave for \(p\leq1\), the assumptions in Theorem \ref{thm:approx} hold and our algorithms apply to these objectives. 

Because the assumptions in Theorem \ref{thm:approx} does not change under an affine transformation of \(g\), we may also take any weighting and introduce additive constants on the square norm. For example, we can take the  average squared norm of the projections rather than the
total, replacing \(\|A_i P\|^2\) by \(\frac 1{m_i}\|A_i P\|^2\) where \(m_i\) is the number of data points in \(A_i\), in any of the discussed fairness criteria. This normalization equalizes  weight of each group, which can be useful when groups are of very different sizes, and is also used in all of our experiments. More generally, one can weight each \(f_i\) by a positive constant \(w_i\), where the appropriate weighting of $k$ objectives often depends on the
context and application. Another example is to replace \(f_i(P)=\|A_i P\|^2\) by \(\|A_i P\|^2 - \|A_i\|^2\), which optimizes the worst reconstruction error rather than the worst variance across all groups  as in the \fpca{} definition.  

\subsection{Marginal loss objective}

We now present  a novel fairness criterion \textit{marginal loss} objective. 
We first give a motivating example of two groups for this objective, shown in Figure~\ref{fig:bad-for-mm}.

In this example, two groups can have very different variances when
projected onto one dimension: the first group has a perfect representation in the
horizontal direction and enjoys high variance, while the second has lower variance for
every projection. Thus, asking for a projection which 
maximizes the minimum variance
might incur loss of variance on the first group while not improving the second group. So, 
minimizing the maximum reconstruction error of these two groups fails to
account for the fact that two populations might have wildly different
representation variances when embedded into $d$ dimensions. Optimal
solutions to such objective might behave in a counterintuitive way, preferring to
exactly optimize for the group with smaller inherent representation
variance rather than approximately optimizing for both groups
simultaneously. We find this behaviour undesirable---it requires
sacrifice in quality for one group for no improvement for
the other group. In other words, it does not satisfy \textit{Pareto-optimality}.
 \begin{figure}[!tbp]
  \centering
  \includegraphics[width=0.5\textwidth]{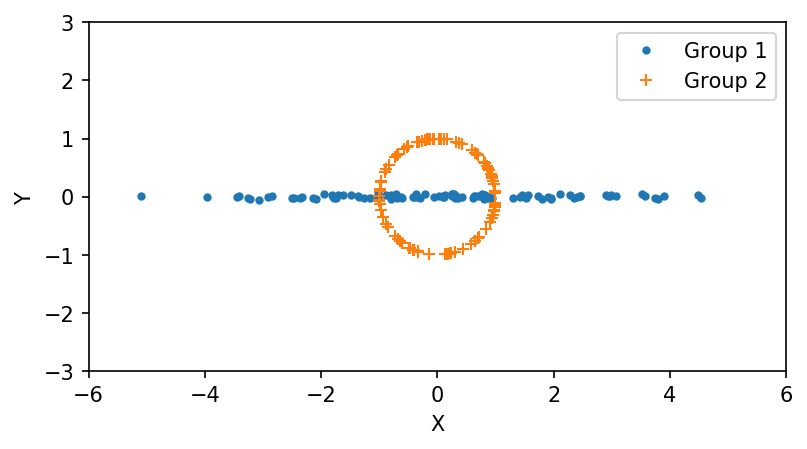}
  \caption{A distribution of two groups where, when projected onto one dimension, maximizing the minimum variance and minimizing the maximum reconstruction error are undesirable objectives.}
\label{fig:bad-for-mm}
\end{figure}

We therefore turn to finding a projection which minimizes the
  maximum deviation of each group from its optimal projection. This
  optimization asks that two groups suffer a similar \emph{decrease}  in variance for
  being projected together onto $d$ dimensions compared to their
  individually optimal projections  ("the marginal cost of sharing a common subspace"). Specifically, we set the utility
$
f_i(P) := \|A_i P\|^2_F - \max_{Q\in \PP_d} \|A_i Q\|^2_F
$
as the change of variance for each group \(i\) and \(g(f_1,f_2,\ldots,f_k):=\min\set{f_1,f_2,\ldots,f_k}\) in the \gfpca{} formulation. This gives an optimization problem
\begin{equation}
\min_{P\in\PP_d} \max_{i\in[k]} \pr{\max_{Q\in \PP_d} \|A_i Q\|^2_F - \|A_i P\|^2_F } \label{eq:marginal-loss}
\end{equation}
 We refer to \(\max_{Q\in \PP_d} \|A_i Q\|^2_F - \|A_i P\|^2\) as the  \textit{loss} of group \(i\) by a projection \(P\), and we call the objective \eqref{eq:marginal-loss} to be minimized the \textit{marginal loss} objective.

For two groups,  marginal loss objective prevents
the optimization from incurring loss of variance for one subpopulation without improving the other as seen in Figure~\ref{fig:bad-for-mm}.
In fact, we show that an optimal solution of marginal loss objective always gives the same loss for two groups.
As a result, marginal loss objective not only satisfies Pareto-optimality\,, but also  equalizes  individual utilities, a property that none of the previously mentioned fairness criteria necessarily holds.  \begin{theorem}
\label{thm:sameLoss}
Let $P^*$ be an optimal solution to \eqref{eq:marginal-loss} for two groups. Then,
\[
 \max_{Q\in \PP_d} \|A_1Q\|^2_F - \|A_1P^*\|^2_F= \max_{Q\in \PP_d} \|A_2Q\|^2_F - \|A_2P^*\|^2_F
\]
\end{theorem}
Theorem~\ref{thm:sameLoss} can be proved by a "local move" argument on the space of all \(d\)-dimensional subspaces equipped with a carefully defined distance metric. We define a new metric space since the move is not valid on the  natural choice of space, namely the domain \(\PP_d\subseteq\R^{n\times n}\) with Euclidean distance, as the set \(\PP_d\) is not convex. We also give a proof using the SDP relaxation of the problem. 
We refer the reader to Appendix \ref{app:proof} for the proofs of~Theorem~\ref{thm:sameLoss}.
In general, Theorem \ref{thm:sameLoss} does not generalize to more than two groups (see Section \ref{sec:gap}).

Another motivation of the marginal loss objective is to check bias in PCA performance on subpopulations. A data set may show a small gap in variances or reconstruction errors of different groups, but a significant gap in losses. An example is the labeled faces in the wild data set (LFW) \cite{LFWTech} where we check both reconstruction errors and losses of male and female groups. As shown in Figure \ref{fig:loss}, the gap between male and female reconstruction errors is about 10\%, while the marginal loss of female is about 5 to 10 times of the male group. This suggests that the difference in marginal losses is a primary source of bias, and therefore marginal losses rather than reconstruction errors should be equalized.

\begin{figure*}[t]
\centering
  \includegraphics[width=0.4\linewidth]{re_face_2.png}
  \hspace{.5cm}
  \includegraphics[width=0.4\linewidth]{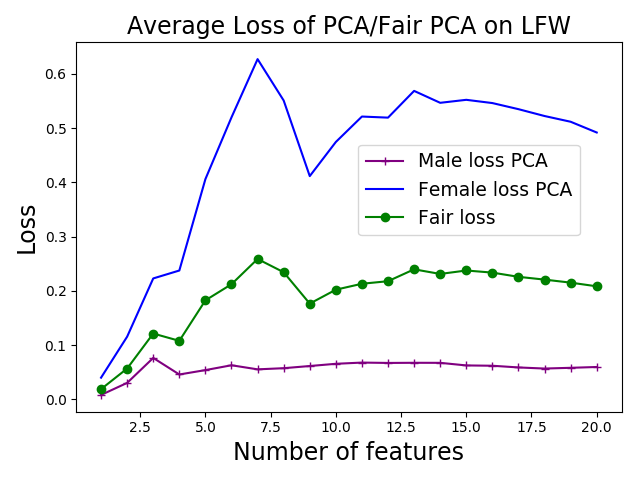}
  \caption{Left: reconstruction error of PCA on labeled faces in the wild data set (LFW), separated by
  gender. Right: marginal loss objective on the same data set. The fair loss is obtained by a solution to the marginal loss objective, which equalizes the two losses.}
  \label{fig:loss-pca}
\end{figure*}


\section{Iterative rounding framework with applications to \fpca{}} \label{sec:approx}


In this section, we  
 give an iterative rounding algorithm and prove Theorem \ref{thm:low-rank-iterative}. The algorithm is specified in Algorithm~\ref{alg:sdp}. The algorithm maintains three subspaces of \(\R^{n\times n}\) that are mutually orthogonal. Let $F_0, F_1, F$ denote matrices whose columns form an orthonormal basis of these subspaces. We will also abuse notation and denote these matrices by sets of vectors in their columns. We let the rank of $F_0, F_1$ and $F$ be $r_0, r_1$ and $r$, respectively. We will ensure that $r_0+r_1+r=n$, i.e., vectors in $F_0, F_1$ and $F$ span $\R^n$. 


We initialize $F_0=F_1=\emptyset$ and $F=I_n$. Over iterations, we  increase the subspaces spanned by columns of $F_0$ and $F_1$ and decrease $F$ while maintaining pairwise orthogonality. The vectors in columns of $F_1$ will be eigenvectors of our final solution with eigenvalue $1$. In each iteration, we project the constraint matrices $A_i$ orthogonal to $F_1$ and $F_0$. We will then formulate a residual SDP using columns of $F$ as a basis and thus the new constructed matrices will have size $r\times r$. To readers familiar with the iterative rounding framework in linear programming, this generalizes the method of fixing certain variables to $0$ or $1$ and then formulating the residual problem.
We also maintain a subset of constraints indexed by $S$ where $S$ is initialized to $\{1,\ldots, m\}$. 

In each iteration, we formulate the following $\sdpr$ with variables $X(r)$ which will be a $r\times r$ symmetric matrix. Recall $r$ is the number of columns in $F$.
\begin{align*}
\max \, & \, \langle F^TCF, X(r)\rangle \\
\langle F^T A_i F, X(r)\rangle & \ge b_i- F_1^T A_i F_1 \quad i \in S\\
\tr(X) & \le d-\rank(F_1)\\
0 \preceq \, X(r)  &\preceq I_r
\end{align*}

\begin{algorithm}[h] 
\caption{Iterative rounding algorithm  \sdpit{}} 
\label{alg:sdp} 
\begin{algorithmic}[1]
    \Statex \textbf{Input:} $C$  a real $n\times n$ matrix,  $\A=\{A_1,\ldots, A_m\}$  a set of real $n \times n$  matrices, \(d\leq n\), and \(b_1,\ldots b_m\in\RR\).    \Statex \textbf{Output:} A feasible solution \(\tilde X\) to \(\sdpa\)   \State Initialize $F_0, F_1$ to be empty matrices and $F\gets I_n$, $S\gets \{1,\ldots, m\}$.

    \State If the $\sdpa$ is infeasible, declare infeasibility and stop.
    \While{ $F$ is not the empty matrix}
    \State  Solve $\sdpr$ to obtain an extreme point $X^*(r)=\sum_{j=1}^r \lambda_j v_j v_j^T$ where $\lambda_j$ are the eigenvalues and $v_j\in \R^r$ are the corresponding eigenvectors.
\State For any eigenvector $v$ of $X^*(r)$ with eigenvalue $0$, let $F_0\gets F_0\cup \{Fv\}.$
\State For any eigenvector $v$ of $X^*(r)$ with eigenvalue $1$, let $F_1\gets F_1\cup \{Fv\}.$
\State Let $X_f=\sum_{j: 0<\lambda_j<1} \lambda_j v_j v_j^T$. If there exists a constraint $i\in S$ such that $\langle F^T A_i F, X_f\rangle < \Delta(\A)$, then
$S\gets S\setminus \{i\}.$ \label{alg:iter-delete}

\State For every eigenvector $v$ of $X^*(r)$ with eigenvalue not equal to $0$ or $1$, consider the vectors $Fv$ and form a matrix with these columns and use it as the new $F$.

\EndWhile
    \State Return $\tilde{X}=F_1F_1^T$. 
\end{algorithmic}
\end{algorithm}

It is easy to see that the semi-definite program remains feasible over all iterations if $\sdpa$ is declared feasible in the first iteration. Indeed the solution $X_f$ defined at the end of any iteration is a feasible solution to the next iteration. 
We also need the following standard claim.

\begin{claim} \label{cl:topl}
Let $Y$ be a positive semi-definite matrix such that $Y\preceq I$ with $\tr(Y)\leq l$. Let $B$ be a real matrix of the same size as $Y$ and let $\lambda_i(B)$ denote the $i^{th}$ largest singular value of $B$. Then
$$\langle B, Y\rangle \leq \sum_{i=1}^l \lambda_i(B).$$
\end{claim}

The following result follows from Corollary~\ref{cor:fractional} and Claim~\ref{cl:topl}. Recall that
$$\Delta(\A):=\max_{S\subseteq [m]} \sum_{i=1}^{\lfloor \sqrt{2|S|}+1\rfloor} \sigma_i(S).$$
where $\sigma_i(S)$ is the $i$'th largest singular value of $\frac{1}{|S|} \sum_{i\in S} A_i$. We let $\Delta$ denote $\Delta(\A)$ for the rest of the section.

\begin{lemma} \label{lem:iter-delete}
Consider any extreme point solution $X(r)$ of $\sdpr$ such that $\rank(X(r))>\tr(X(r))$. Let $X(r)=\sum_{j=1}^r \lambda_j v_j v_j^T$ be its eigenvalue decomposition and $X_f=\sum_{0<\lambda_j<1} \lambda_j v_j v_j^T$. Then there exists a constraint $i$ such that $\langle F^T A_i F, X_f\rangle <  \Delta$.
\end{lemma}
\begin{proof}
Let $l=|S|$. From Corollary~\ref{cor:fractional}, it follows that the number of fractional eigenvalues of $X(r)$ is at most $ -\frac{1}{2}+\sqrt{2l+\frac{9}{4}}\leq  \sqrt{2{l}}+1$.
Observe that $l>0$ since $\rank(X(r))>\tr(X(r))$. Thus, we have $\rank(X_f)\leq \sqrt{2l}+1$. Moreover, $0\preceq X_f\preceq I$,  so from Claim~\ref{cl:topl}, we obtain that

$$\left\langle\sum_{j\in S} F^T A_j F, X_f\right\rangle  \leq \sum_{i=1}^{\lfloor \sqrt{2l}+1\rfloor} \sigma_i\left( \sum_{j\in S} F^T A_j F\right) \leq \sum_{i=1}^{\lfloor \sqrt{2l}+1\rfloor} \sigma_i\left( \sum_{j\in S} A_j \right) \leq l \cdot \Delta  $$
where the first inequality follows from Claim~\ref{cl:topl} and the second inequality follows since the sum of top $l$ singular values reduces after projection.
But then we obtain, by averaging, that there exists $j\in S$  such that
$$\langle  F^T A_j F, X_f\rangle < \frac{1}{l}\cdot l\Delta=\Delta$$
as claimed.
\end{proof}

Now we complete the proof of Theorem~\ref{thm:low-rank-iterative}. Observe that the algorithm always maintains that at the end of each iteration,  $\tr(X_f)+\rank(F_1)\leq d$.  Thus at the end of the algorithm, the returned solution has rank at most $d$. 
Next, consider the solution $X=F_1 F_1^T + FX_fF^T$ over the course of the algorithm. Again, it is easy to see that the objective value is non-increasing over the iterations. This follows since $X_f$ defined at the end of an iteration is a feasible solution to the next iteration. 
%

Now we argue a bound on the violation in any constraint $i$. While the constraint $i$ remains in the SDP, the solution $X=F_1F_1^T + FX_fF^T$ satisfies
\begin{align*}
\langle A_i, X \rangle &= \langle A_i, F_1 F_1^T\rangle + \langle A_i, F X_f F^T\rangle \\
&= \langle A_i, F_1 F_1^T\rangle + \langle F^T A_i F,  X_f \rangle \leq \langle A_i, F_1 F_1^T\rangle + b_i-  \langle A_i, F_1 F_1^T\rangle = b_i.
\end{align*}
where the inequality again follows since $X_f$ is feasible with the updated constraints.

When constraint $i$ is removed, it might be violated by a later solution. At this iteration,
$\langle F^TA_iF, X_f \rangle \le \Delta$. Thus,
$\langle A_i, F_1F_1^T\rangle \ge b_i - \Delta$. In the final solution, this bound can only go up as $F_1$ might only become larger. This completes the proof of the theorem.

We now analyze the runtime of the algorithm which contains at most \(m\) iterations. First we note that we may avoid computing \(\Delta(\A)\) by deleting  a constraint \(i\) from \(S\) with smallest \(\langle F^T A_i F, X_f\rangle\) instead of checking \(\langle F^T A_i F, X_f\rangle<\Delta(\A)\) in  step \eqref{alg:iter-delete} of Algorithm \ref{alg:sdp}. The guarantee still holds by Lemma \ref{lem:iter-delete}. Each iteration requires solving an SDP and eigenvalue decompositions over \(r\times r\) matrices, computing \(F_0,F_1,F\), and finding \(i\in S\) with the smallest \(\langle F^T A_i F, X_f\rangle\). These can be done in \(O(r^{6.5})\), \(O(r^2n)\), and \(O(rmn^2)\) time. However, the result in Section \ref{sec:sdp-round} shows that after solving the first \(\sdpr\), we have \(r\leq O(\sqrt m)\), and hence the total runtime of iterative rounding after solving for an extreme solution of the SDP relaxation) is \(O(m^{4.25}+m^{1.5}n^2)\).

\paragraph{Application to  \fpca{}.}
For  \fpca{}, iterative rounding recovers a rank-$d$ solution whose variance goes down from the SDP solution by at most $\Delta(\{A_1^TA_1,\ldots, A_k^TA_k\})$. While this is no better than what we get by scaling (Corollary \ref{cor:approx-SDP2}) for the max variance objective function, when we consider the marginal loss, i.e., the difference between the variance of the common $d$-dimensional solution and the best $d$-dimensional solution for each group, then iterative rounding can be much better. The scaling solution guarantee relies on the max-variance being a concave function, and for the marginal loss, the loss for each group could go up proportional to the {\em largest} max variance (largest sum of top $k$ singular values over the groups).  With iterative rounding applied to the SDP solution, the loss $\Delta$ is the sum of only $O(\sqrt{k})$ singular values of the average of some subset of data matrices, so  it can be better by as much as a factor of $\sqrt{k}$.
\section{Polynomial time  algorithm for fixed number of groups}\label{sec:poly-time}

\paragraph{Functions of quadratic maps.} We briefly summarize the approach of \cite{grigoriev2005polynomial}.
Let $f_1,\ldots,f_k:\R^n \rightarrow \R$ be real-valued quadratic functions in $n$ variables. Let $p:\R^k \rightarrow \R$ be a polynomial of degree $\ell$ over some subring of $\R^k$ (e.g., the usual $(\times, +)$ or $(+,\min)$) The problem is to find all roots of the polynomial $p(f_1(x),f_2(x),\ldots,f_k(x))$, i.e.,  the set
\[
Z = \{x\, :\, p(f_1(x),f_2(x),\ldots,f_k(x))=0\}.
\]
First note that the set of solutions above is in general not finite and is some manifold and highly non-convex.
The key idea of Grigoriev and Paleshnik (see also Barvinok \cite{barvinok1993feasibility} for a similar idea applied to a special case) is to show that this set of solutions can be partitioned into a relatively small number of connected components such that there is an into map from these components to roots of a univariate polynomial of degree $(\ell n)^{O(k)}$; this therefore bounds the total number of components. The proof of this mapping is based on an explicit decomposition of space with the property that if a piece of the decomposition has a solution, it must be the solution of a linear system. The number of possible such linear systems is bounded as $n^{O(k)}$, and these systems can be enumerated efficiently.

The core idea of the decomposition starts with the following simple observation that relies crucially on the maps being quadratic (and not of higher degree).
\begin{proposition}
The partial derivatives of any degree $d$ polynomial $p$ of quadratic forms $f_i(x)$, where $f_i:\R^n \rightarrow \R$, is linear in $x$ for any fixed value of $\{f_1(x), \ldots, f_k(x)\}$.
\end{proposition}

To see this, suppose $Y_j = f_j(x)$ and write
\[
\frac{\partial p}{\partial x_i} = \sum_{j=1}^k \frac{\partial p(Y_1, \ldots, Y_k)}{\partial Y_j} \frac{\partial Y_j}{\partial x_i} = \sum_{j=1}^k
\frac{\partial p(Y_1, \ldots, Y_k)}{\partial Y_j}\frac{\partial f_j(x)}{\partial x_i}.
\]
Now the derivatives of $f_j$ are linear in $x_i$ as $f_j$ is quadratic, and so for any fixed values of $Y_1, \ldots, Y_k$, the expression is linear in $x$.

The next step is a nontrivial fact about connected components of analytic manifolds that holds in much greater generality. Instead of all points that correspond to zeros of $p$, we look at all ``critical" points of $p$ defined as the set of points $x$ for which the partial derivatives in all but the first coordinate, i.e.,
\[
Z_c = \{x \,:\, \frac{\partial p}{\partial x_i} = 0, \quad \forall 2 \le i  \le n\}.
\]
The theorem says that $Z_c$ will intersect every connected component of $Z$ \citep{grigor1988solving}.

Now the above two ideas can be combined as follows. We will cover all connected components of $Z_c$. To do this we consider, for each fixed value of $Y_1, \ldots, Y_k$, the possible solutions to the linear system obtained, alongside minimizing $x_1$. The rank of this system is in general at least $n-k$ after a small perturbation (while \cite{grigoriev2005polynomial} uses a deterministic perturbation that takes some care, we could also use a small random perturbation). So the number of possible solutions grows only as exponential in $O(k)$ (and not $n$), and can be effectively enumerated in time $(\ell d)^{O(k)}$. This last step is highly nontrivial, and needs the argument that over the reals, zeros from distinct components need only to be computed up to finite polynomial precision (as rationals) to keep them distinct. Thus, the perturbed version still covers all components of the original version. In this enumeration, we check for true solutions. The method actually works for any level set of $p$, $\{x\, :\, p(x)=t\}$ and not just its zeros. With this, we can optimize over $p$ as well. We conclude this section by paraphrasing the main theorem from \cite{grigoriev2005polynomial}.

\begin{theorem}\label{thm:GP}\citep{grigoriev2005polynomial}
Given $k$ quadratic maps $q_1,\ldots, q_k:\R^k \rightarrow \R$ and a polynomial $p:\R^k \rightarrow \R$ over some computable subring of $\R$ of degree at most $\ell$, there is an algorithm to compute a set of points satisfying $p(q_1(x),\ldots,q_k(x))=0$ that meets each connected component of the set of zeros of $p$ using at most $(\ell n)^{O(k)}$ operations with all intermediate representations bounded by $(\ell n)^{O(k)}$ times the bit sizes of the coefficients of $p,q_1, \ldots, q_k$. The minimizer, maximizer or infimum of any polynomial $r(q_1(x), \ldots, q_k(x))$ of degree at most $\ell$ over the zeros of $p$ can also be computed in the same complexity.
\end{theorem}

\subsection{Proof of Theorem~\ref{thm:fixed-kd}}
We apply Theorem~\ref{thm:GP} and the corresponding algorithm as follows. Our variables will be the entries of an $n \times d$ matrix $P$. The quadratic maps will be $f_i(P)$ plus additional maps for $q_{ii}(P) = \|P_i\|^2-1$ and $q_{ij}(P)=P_i^TP_j$ for columns $P_i, P_j$ of $P$. The final polynomial is
\[
p(f_1, \ldots, f_k,  q_{11}, \ldots, q_{dd}) = \sum_{i\le j} q_{ij}(P)^2.
\]
We will find the maximum of the polynomial $r(f_1, \ldots f_k) = g(f_1,\ldots, f_k)$ over the set of zeros of $p$ using the algorithm of Theorem~\ref{thm:GP}. Since the total number of variables is $dn$ and the number of quadratic maps is $k+ d(d+1)/2$, we get the claimed complexity of $O(\ell d n)^{O(k+d^2)}$ operations and this times the input bit sizes as the bit complexity of the algorithm.

\section{Hardness and integrality gap} \label{sec:hardness-and-gap}
\subsection{NP-Hardness}
In this section, we show NP-hardness of \fpca{} even for \(d=1\), proving Theorem \ref{thm:NPhard}.
\begin{theorem}\label{thm:nphardness}
The \fpca{} problem:
\begin{align}
\max_{z\in\R,P\in\R^{n \times d}} &  z \textup{\qquad \quad subject to} \label{eq:fair-hard1}\\
\an{B_i,PP^T} &\geq z \qquad ,\forall i\in[k] \label{eq:fair-hard2} \\
P^T P= I_d \label{eq:fair-hard3}
\end{align}
for arbitrary \(n\times n\) symmetric real PSD matrices \(B_1,\ldots,B_k\) is NP-hard for \(d=1\) and \(k=O(n)\).
\end{theorem}

\begin{proofof}{Theorem~\ref{thm:nphardness}}
We reduce another NP-hard problem  MAX-CUT to \fpca{} with \(d=1\). In MAX-CUT, given a simple graph \(G=(V,E)\), we optimize
\begin{align}
\max_{S \subseteq V} e(S,V\setminus S)
\end{align}
over all subset \(S\) of vertices. Here, \(e(S,V\setminus S)=|\set{e_{ij}\in E: i\in S, j \in V\setminus S}|\) is the size of the cut \(S\) in \(G\).
As common in NP-hard problems, the decision version of MAX-CUT:
\begin{equation}
\exists? S\subseteq V:e(S,V\setminus S) \geq b
\end{equation}
for an arbitrary \(b>0\) is also NP-hard. We may write MAX-CUT as an integer program as follows:
\begin{align}
\exists?v\in \set{-1,1}^V: \ &\frac{1}{2}\sum_{ij\in E} \pr{1-v_i v_j} \geq b.      \label{eq:max-cut-obj}
\end{align}
Here \(v_i\) represents whether a vertex \(i\) is in the set \(S\) or not:
\begin{equation}
v_i=\begin{cases}1 & i \in S \\
-1 & i\notin S \\
\end{cases},
\end{equation}
and it can be easily verified that the objective represents the desired cut function.

We now show that this MAX-CUT integer feasibility problem can be formulated as an instance of \fpca{} \eqref{eq:fair-hard1}-\eqref{eq:fair-hard3}. In particular, it will be formulated as a feasibility version of \fpca{} by checking if the optimum \(z\) of \fpca{} is at least \(b\). We choose \(d=1\) and \(n=|V|\) for this instance, and we write \(P=[u_1;\ldots;u_n]\in\R^n\). The rest of the proof is to show that it is possible to construct constraints in the form \eqref{eq:fair-hard2}-\eqref{eq:fair-hard3} to 1) enforce a discrete condition on \(u_i\) to take only two values, behaving similarly as \(v_i\); and 2) check an objective value of MAX-CUT.

Note that constraint \eqref{eq:fair-hard3} requires \(\sum_{i=1}^n {u_i}^2=1\) but \(\sum_{i=1}^n {v_i}^2=n\). Hence, we  scale the variables in MAX-CUT problem by writing \(v_i=\sqrt{n}u_i\) and rearrange terms in \eqref{eq:max-cut-obj} to obtain an equivalent formulation of MAX-CUT:
\begin{align}
\exists?u\in \set{-\frac{1}{\sqrt{n}},\frac{1}{\sqrt{n}}}^n: \ &n\sum_{ij\in E} -u_iu_j \geq 2b-|E|      \label{eq:max-cut-obj2}
\end{align}

We are now ready to give an explicit construction of \(\set{B_i}_{i=1}^k\) to solve MAX-CUT formulation \eqref{eq:max-cut-obj2}. Let \(k=2n+1\). For each \(j=1,\ldots,n\),  define
\begin{displaymath}
B_{2j-1} =bn\cdot \diag(\mathbf{e_j}), B_{2j} = \frac{bn}{n-1}\cdot \diag(\mathbf{1}-\mathbf{e_j})
\end{displaymath}
where \(\mathbf{e_j}\)  and \(\mathbf{1}\) denote vectors of length \(n\) with all zeroes except one at the \(j\)th coordinate, and with all ones, respectively. It is clear that \(B_{2j-1},B_{2j}\) are PSD. Then for each \(j=1\ldots,n\), the constraints \(\an{B_{2j-1},PP^T}\geq b\) and \(\an{B_{2j},PP^T}\geq b\) are equivalent to
\begin{align*}
u_j^2\geq \frac{1}{n},\text{ and } \sum_{i\neq j} u_j^2 \geq \frac{n-1}{n}
\end{align*}
respectively. Combining these two inequalities with \(\sum_{i=1}^n u_i^2=1\) forces both inequalities to be equalities, implying that \(u_j\in\set{-\frac{1}{\sqrt{n}},\frac{1}{\sqrt{n}}}\) for all \(j\in[n]\), as we aim.

Next, we set
\begin{displaymath}
B_{2n+1}=\frac{bn}{2b-|E|+n^2}\cdot \pr{nI_n-A_G}
\end{displaymath}
where \(A_G=(\mathbb{I}[ij\in E])_{i,j\in[n]}\) is the adjacency matrix of the graph \(G\). Since the matrix \(nI_n - A_G\) is diagonally dominant and real symmetric, \(B_{2n+1}\) is PSD. We have that \(\an{B_{2n+1},PP^T}\geq b\) is equivalent to
\begin{align*}
\frac{bn}{2b-|E|+n^2} \pr{n\sum_{i=1}^n u_i^2-\sum_{ij\in E} u_iu_j } \geq b
\end{align*}
which, by \(\sum_{i=1}^n u_i^2 =1\), is further equivalent to
\begin{displaymath}
n\sum_{ij\in E} -u_iu_j \geq 2b-|E|,
\end{displaymath}
matching \eqref{eq:max-cut-obj2}.
To summarize, we constructed \(B_1,\ldots,B_{2n+1}\) so that checking whether an objective of \fpca{} is at least \(b\) is equivalent to checking whether a graph \(G\) has a cut of size at least \(b\), which is NP-hard.
\end{proofof}


\subsection{Integrality gap} \label{sec:gap}
We showed that \fpca{} for \(k=2\) groups can be solved up to optimality in polynomial time using an SDP. For $k > 2$, we used a different, non-convex approach to get a polynomial-time algorithm for any fixed $k,d$. We show that the SDP relaxation of \fpca{}  has a gap even for \(k=3\) and \(d=1\) in the following lemma.
Here, the constructed matrices $B_i$'s are also PSD, as  required by \(B_i=A_i^TA_i\) for a data matrix \(A_i\) in the \fpca{} formulation. A similar result on tightness of rank violation for larger \(k\) using real (non-PSD) matrices \(B_i\)'s is stated in Lemma \ref{lem:low-rank-tight} in Appendix \ref{sec:tight}.

\begin{lemma}\label{lem:integralityGap}
The \fpca{} SDP relaxation:{\begin{align*}
  \qquad &\max \, \, z \\
\langle B_i, X\rangle & \ge z \quad i \in \{1,\ldots, k\}\\
\tr(X) & \le d\\
0 \preceq \, X  &\preceq I
\end{align*}}
for \(k=3\), \(d=1\), and arbitrary  PSD \(\set{B_i}_{i=1}^k\) contains a gap, i.e. the optimum value of the SDP relaxation is different from one of exact \fpca{} problem.
\end{lemma}

\begin{proofof}{Lemma~\ref{lem:integralityGap}}
Let \(B_1=\begin{bmatrix}2 & 1 \\
1 & 1 \\
\end{bmatrix},
B_2=\begin{bmatrix}1 & 1 \\
1 & 2 \\
\end{bmatrix},
B_3=\begin{bmatrix}2 & -1 \\
-1 & 2 \\
\end{bmatrix}\). It can be checked that \(B_i\) are PSD. The optimum of the relaxation is \(7/4\) (given by the optimal solution \(X=\begin{bmatrix}1/2 & 1/8 \\
1/8 & 1/2 \\
\end{bmatrix}\)).
However, an optimal exact \fpca{} solution is \(\hat{X}=\begin{bmatrix}{16}/{17} & 4/17 \\
4/17 & 1/17 \\
\end{bmatrix}\)
which gives an optimum  \(26/17\) (one way to solve for optimum rank-1 solution \(\hat{X}\) is by parameterizing \(\hat{X}=v(\theta)v(\theta)^T\) for \(v(\theta)=[\cos \theta; \sin \theta]\), \(\theta\in[0,2\pi)\)).
\end{proofof}

The idea of the example in Lemma \ref{lem:integralityGap} is that an optimal solution is close to picking the first axis as a projection, whereas a relaxation solution splits two halves for each of the two axes. Note that the example also shows a gap for marginal loss objective \eqref{eq:marginal-loss}. Indeed, the numerical computation shows that optimal marginal loss (which is to be minimized)\ for the exact problem is 1.298 by \(X^*\approx\begin{bmatrix}0.977 & 0.149 \\
0.149 & 0.023 \\
\end{bmatrix}\) and for relaxed problem is 1.060 by \(\widehat X \approx \begin{bmatrix}0.5 & 0.030 \\
0.030 & 0.5 \\
\end{bmatrix}\). This shows that equal losses for two groups in Theorem \ref{thm:sameLoss} cannot be extended to more than two groups.  The same example also show a gap if the objective is to minimize the maximum reconstruction errors. Gaps for all three objectives remain even after normalizing the data by \(B_i\leftarrow\frac{B_i}{\tr(B_i)}\) by numerical computation. The pattern of  solutions across three objectives and in both unnormalized and normalized setting remains the same as mentioned: an exact solution is a projection close to the first axis, and the relaxation splits two halves for the two axes, i.e., picking \(X\) close to \(I_2\). 

\section{Experiments} \label{sec:experiments}
\sloppy 
\subsection{Efficacy of our algorithms to fairness}
We perform experiments using the algorithm as outlined in Section~\ref{sec:sdp-round} on the Default Credit data set~\citep{default-dataset} for different target dimensions \(d\), and evaluate the fairness performance based on marginal loss and NSW criteria (see Table \ref{tab:fair-criteria} in Section \ref{sec:fairness-criteria} on definitions of these criteria). 
The data consists of 30K data points in 23 dimensions, partitioned into \(k=4,6\) groups by education and gender, and then preprocessed to have mean zero and same variance over features. Our algorithms are set to optimize on either the marginal loss  and NSW objective. The code is publicly available at \url{https://github.com/uthaipon/multi-criteria-dimensionality-reduction}.

\begin{figure}[th]
\centering
\includegraphics[width=0.4\textwidth]{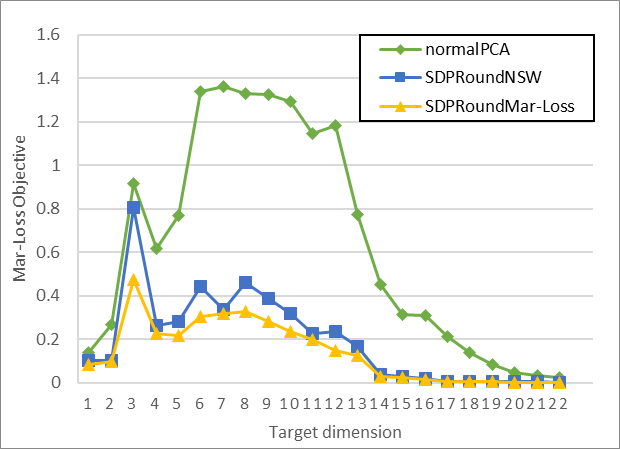}
\includegraphics[width=0.4\textwidth]{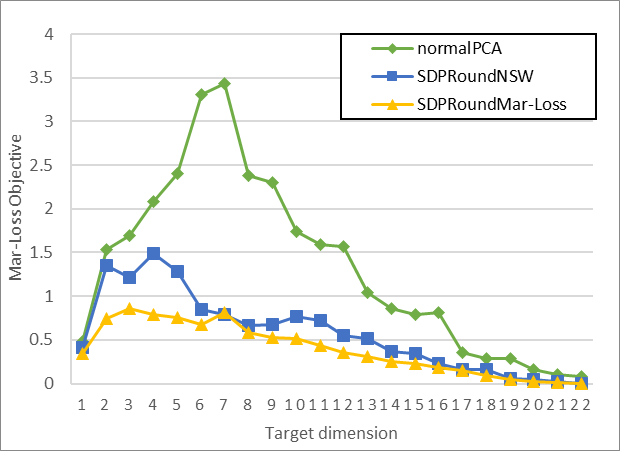}
\caption{Marginal loss function of standard PCA compared to our  SDP-based algorithms on Default Credit data. SDPRoundNSW and SDPRoundMar-Loss are two runs of the SDP-based algorithms maximizing NSW  and minimizing marginal loss. Left: \(k=4\) groups. Right: \(k=6\).}\label{fig:loss}
\end{figure}


Figure \ref{fig:loss} shows  the marginal loss by our algorithms compared to  standard PCA on the entire data set. Our algorithms significantly reduce disparity of marginal loss of PCA that the standard PCA subtly introduces.
We  also assess the performance of PCA with NSW objective, summarized in Figure \ref{fig:NSW}. With respect to NSW, standard PCA performs marginally worse (about 10\%) compared to our algorithms. It is worth noting from Figures \ref{fig:loss} and \ref{fig:NSW} that our algorithms which try to optimize either marginal loss  or NSW also perform well on the other fairness objective, making these PCAs promising candidates for fairness applications.

\begin{figure}[h]
\centering
\includegraphics[width=0.45\textwidth]{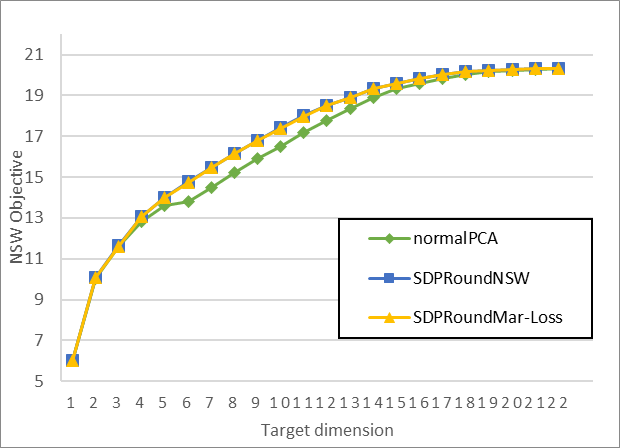}
\includegraphics[width=0.45\textwidth]{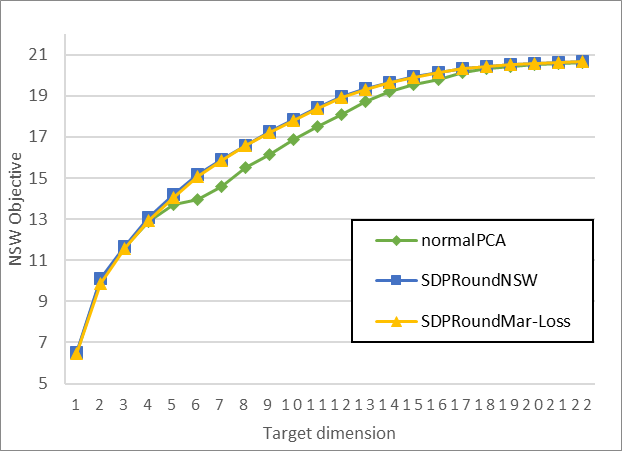}
\caption{NSW objective of  standard PCA compared to our  SDP-based algorithms on Default Credit data. SDPRoundNSW and SDPRoundMar-Loss are two runs of the SDP-based algorithms maximizing NSW objective and minimizing  marginal loss. Left: \(k=4\) groups. Right: \(k=6\).}\label{fig:NSW}
\end{figure}

Same experiments are done on the Adult Income data \citep{adult-data}. Some categorial features are preprocessed into integers vectors and some  features and rows with missing values are discarded. The final preprocessed data contains $m=32560$ data points in $n=59$ dimensions and is partitioned into $k=5$ groups based on race. Figure \ref{fig:Inc} shows the performance of our SDP-based algorithms compared to  standard PCA on marginal loss and NSW objectives. Similar to  Credit Data, optimizing for either marginal loss or NSW gives a PCA solution that also performs well in another criterion and performs better than the standard PCA in both objectives. 
\begin{figure}[h]
\centering
\includegraphics[width=0.45\textwidth]{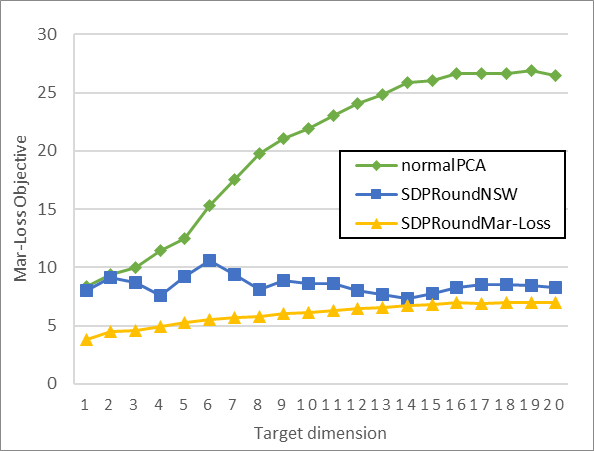}
\includegraphics[width=0.45\textwidth]{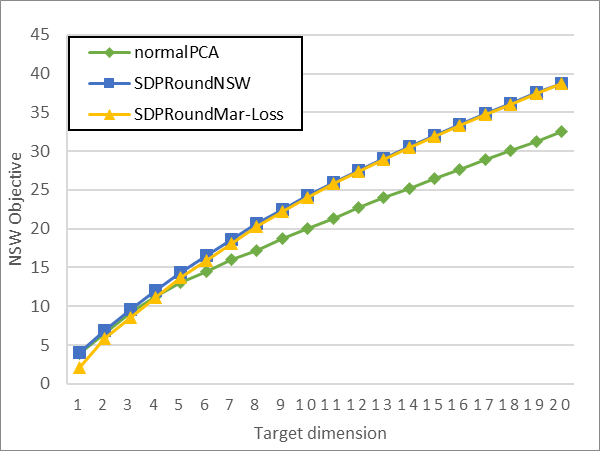}
\caption{Marginal loss and NSW objective of standard PCA compared to our  SDP-based algorithms on Adult Income data. SDPRoundNSW and SDPRoundMar-Loss are two runs of the SDP algorithms maximizing NSW objective and minimizing maximum marginal loss. }\label{fig:Inc}
\end{figure}

We note  details of the labeled faces in the wild data set (LFW) \citep{LFWTech} used in Figures~\ref{fig:wpca}  and \ref{fig:loss-pca} here.
The original data are in 1764 dimensions (42$\times$42 images). We preprocess all data to have mean zero and we normalize each pixel value by multiplying $\tfrac{1}{255}$.  The gender
information for LFW was taken from~\citet{afifi2017afif4}, who
manually verified the correctness of these labels. To obtain the fair loss in Figure  \ref{fig:loss-pca}, we solve using our SDP-based algorithm which, as our theory suggests, always give an exact optimal solution.

\paragraph{Rank violations in experiments.}
In all of the experiments,  extreme point solutions from SDPs enjoy lower rank violation than our worst-case guarantee. Indeed, while the guarantee is that the numbers of additional rank are at most \(s=1,2\) for \(k=4,6\), almost all SDP solutions have   \textit{exact} rank, and in rare cases when the solutions are not exact, the rank violation is only one. As a result, we solve \gfpca{} in practice by solving the SDP relaxation targeting dimension \(d\). If the solution is exact, then we are done. Else, we target dimension \(d-1\) and check if the solution is of rank at most \(d\). If not, we continue to target dimension \(d-2,d-3,\ldots\) until the solution of the SDP relaxation has rank at most \(d\).
While our rank violation guarantee  cannot be improved in general (due to the integrality gap in  Section \ref{sec:gap}; also see Lemma \ref{lem:low-rank-tight} for tightness of the rank violation bound), this opens a question whether the guarantee is better for instances that arise in practice.

\paragraph{Extreme property of SDP relaxation solutions in real data sets.}
One concern for solving an SDP is whether the solver will not return an extreme solution; if so, the SDP-based algorithm requires an additional time to round the solution. 
We found that a solution from SDP relaxation is, in fact, always already extreme in practice. This is because with probability one over random data sets,  SDP is not degenerate, and hence have a unique optimal solution. Since any linear optimization over a compact, convex set must have an extreme optimal solution, this optimal solution is necessarily extreme. Therefore, in practice, it is not necessary to apply the SDP rounding algorithm to the solution of SDP relaxation.
As an application, any faster algorithm or heuristics  which can solve SDP relaxation to optimality in practice will  always obtain a low-rank solution immediately.
We discuss some useful  heuristics in Section \ref{sec:heuristics}.

\subsection{Runtime improvement} \label{sec:heuristics}
We found that the running time of solving SDP, which depends on \(n\), is the bottleneck in all  experiments. Each run (for one value of \(d\)) of the experiments is fast ($< 0.5$ seconds) on Default Credit data (\(n=23\)), whereas a run on Adult Income data (\(n=59\)) takes between 10 and 15 seconds on a personal computer. The runtime is not noticeably impacted by the numbers of data points and groups: larger \(m\) only increases the data preprocessing time (the matrix multiplication \(B_i=A_i^TA_i\)) to obtain \(n\times n\) matrices, and larger  \(k\) simply increases the number of constraints. SDP solver and rounding algorithms can handle moderate number of affine constraints efficiently. This observation is as expected from the theoretical analysis.

In this section, we show two heuristics for solving the SDP relaxation that run significantly faster in practice for large data sets: multiplicative weight update (MW) and Frank-Wolfe (FW). We also discuss several findings and suggestions for implementing our algorithms in practice.
Both heuristics are publicly available at the same location as SDP-based algorithm experiments.

For the rest of this section, we assume that the utility of each group is  \(u_i(X)=\an{B_i,X}\) for real \(B_i\in\R^{n\times n}\), and that \(g(z_1,\ldots,z_k)\) is a concave function of \(z_1,\ldots,z_k\). When \(u_i\) is other linear function in \(X\), we can model such different utility function by modifying \(g\) without changing the concavity of \(g\). The SDP relaxation of \gfpca{} can then be framed as the following SDP.
\begin{align}
  \max_{X\in\R^{n \times n}} &g(z_1,z_2,\ldots,z_k) \text{ subject to }
  \label{eq:sdp-gen-obj-in-z}\\
  z_i&=\an{B_i,X} \qquad \forall i=1,2,\ldots,k \label{eq:sdp-gen-variance-con}\\
   \tr(X)&\leq d\\
   0\preceq X&\preceq I_n \label{eq:sdp-gen-con-in-z}
\end{align}
\subsubsection{Multiplicative Weight Update (MW)}

One alternative method to solving \eqref{eq:sdp-gen-obj-in-z}-\eqref{eq:sdp-gen-con-in-z} is multiplicative weight (MW) update \citep{arora2012multiplicative}.
Though this algorithm has theoretical guarantee, in practice the learning rate is tuned more aggressively and the algorithm becomes a heuristic without any certificate of optimality. We show the primal-dual derivation of MW which provides the primal-dual gap to certify optimality.

We take the Lagrangian  with dual constraints in \eqref{eq:sdp-gen-variance-con} to obtain
that the optimum of the SDP equals to
\begin{align*}
 \max_{\substack{X\in\R^{n \times n}\\z\in\R^n\\ \tr(X)=d\\0\preceq X\preceq I}} \inf_{w\in\R^k} g(z) + \sum_{i=1}^k w_i\pr{\an{B_i,X}-z_i} 
\end{align*}
By strong duality, we may swap \(\max\) and \(\inf\). After rearranging, the optimum of the SDP equals
\begin{align}
\inf_{w\in\R^k} \pr{\max_{\substack{X\in\R^{n \times n}\\ \tr(X)=d,0\preceq X\preceq I}}\sum_{i=1}^k w_i\an{B_i,X} - \min_{z\in \R^n} \pr{w^T z - g(z)}} \label{eq:sdp-dual-raw}
\end{align}
The inner optimization 
\begin{align}
\max_{\substack{X\in\R^{n \times n}\\ \tr(X)=d,0\preceq X\preceq I}}\sum_{i=1}^k w_i\an{B_i,X} \label{eq:sdp-dual-to-primal-solution}
\end{align}
in \eqref{eq:sdp-dual-raw} can easily be computed by standard PCA on weighted data \(\sum_{i=1}^k w_i\cdot B_i\) projecting from \(n\) to \(d\) dimensions. The term \eqref{eq:sdp-dual-to-primal-solution} is also convex in \(w\), as it is a maximum of (infinitely many) linear functions. The term \(\min_{z\in \R^n} \pr{w^T z - g(z)}\) is also known as concave conjugate of \(g\), which we will denote by \(g_*(w)\). Concave conjugate \(g_*(w)\) is concave, as it is a minimum of linear functions. Hence, \eqref{eq:sdp-dual-raw} is a convex optimization problem.

Solving \eqref{eq:sdp-dual-raw} depends on the form of \(g_*(w)\). For each fairness criteria outlined in this paper, we summarize  the form of \(g_*(w)\) below.
\begin{description}
\item[Max-Min Variance (\fpca{} or MM-Var)]: the fairness objective $g(z) = \min_{i\in[k]} z_i$ gives
\begin{align*}
g_*(w) = \begin{cases}0 & \text{if } w\geq 0,\sum_{i=1}^k w_i = 1 \\
-\infty & \text{otherwise} \\
\end{cases}
\end{align*}

\item[Min-Max Loss (MM-Loss)]: the fairness objective (recall  \eqref{eq:marginal-loss}) $g(z) = \min_{i\in[k]}  z_i-\beta_i$, where $\beta_i=\max_{Q\in \PP_d} \|A_i Q\|^2_F$ is the optimal variance of the group \(i\), gives
\begin{align*}
g_*(w) = \begin{cases}\sum_{i=1}^k w_i\beta_i & \text{if } w\geq 0,\sum_{i=1}^k w_i = 1 \\
-\infty & \text{otherwise} \\
\end{cases}
\end{align*}
More generally, the above form of \(g_*(w)\) holds for any constants \(\beta_i\)'s. For example, this calculation also captures min-max reconstruction error: \(g(X)=\min_{i\in[k]} \set{-\norm{A_i-A_iP}_{F}^2}=\min_{i\in[k]} \{z_i-\tr(B_i)\}\) (recall that \(X=PP^T\), \(B_i=A_i^T A_i\), and \(z_i=\an{B_i,X}\)). 

\item[Nash Social Welfare (NSW)]: the fairness objective $g(z) = \sum_{i=1}^k \log (z_i)$ gives
\begin{align*}
g_*(w) = \begin{cases}\sum_{i=1}^k (1+\log w_i) & \text{if } w> 0 \\
-\infty & \text{otherwise} \\
\end{cases}
\end{align*}
\end{description}  

For fairness criteria of the "max-min" type, such as MM-Var and MM-Loss, solving the dual problem is an optimization over a simplex with  standard PCA as the function evaluation oracle. This can be done using mirror descent \cite{nemirovsky1983problem} with negative entropy potential function \(R( w)=\sum_{i=1}^k w_i\log w_i\). The algorithm is identical to multiplicative weight update (MW) by \citep{arora2012multiplicative}, described in Algorithm \ref{alg:mw}, and the  convergence bounds from mirror descent and \citep{arora2012multiplicative}  are identical. However, with primal-dual formulation, the dual solution \(w_i\) obtained in each step of mirror descent can be used to calculate the dual objective in \eqref{eq:sdp-dual-raw}, and the optimum \(X\) in \eqref{eq:sdp-dual-to-primal-solution} is used to calculate the primal objective, which gives the duality gap. The algorithm runs iteratively until the duality gap is less than a set threshold. We summarize MW in Algorithm \ref{alg:mw}.

\begin{algorithm}[h] 
\caption{Multiplicative weight update (MW) for \gfpca{}} 
\label{alg:mw} 
\begin{algorithmic}[1]
    \Statex \textbf{Input:} PSD \(B_1,\ldots,B_k\in\R^{n\times n}\), \(\beta_1,\ldots,\beta_k\in\R\), learning rate \(\eta>0\), accuracy goal \(\epsilon\).
    \Statex \textbf{Output:} an \(\epsilon\)-additive approximate solution \(\hat X\) to 
\[\max_{\substack{X\in\R^{n \times n}\\ \tr(X)=d,0\preceq X\preceq I}}\ \pr{g(X):= \min_{i\in[k]}  \an{B_i,X}-\beta_i }\]
That is, \(g^*- g(\hat X) \leq \epsilon\) where \(g^*\) is the optimum of the above maximization.
    \State Initialize \(w^{0}\leftarrow (1/k,\ldots,1/k)\in\R^k\); initialize dual bound \(y^{(0)}\leftarrow \infty\)
    \State \(t\leftarrow 0\)
    \While{true}
    \State \(P_t \leftarrow VV^T\) where \(V\) is  \(n\)-by-\(d\) matrix of top \(d\) eigenvectors of \(\sum_{i\in[k]}w_iB_i\) 
\Comment oracle of MW 
    \State \(y_i^{(t)}\leftarrow\an{B_i,P_t}-\beta_i\) for \( i=1,\ldots,k \)
    \State
    \(\hat{w}_i^{(t)}\leftarrow w_i^{(t-1)}e^{-\eta y_i^{(t)}}\) for \(i=1,\ldots,k\)
    \State
    \(w_i^{(t)}\leftarrow \hat{w}_i^{(t)}/(\sum_{i\in[k]}\hat{w}_i^{(t)})\) for \(i=1,\ldots,k\)
    \Statex \(\triangleright\) Compute the duality gap
    \State \(X_t \leftarrow \frac 1t\sum_{s\in[t]}P_s\)
    \State \(y^{(t)} \leftarrow \min\set{y^{(t-1)},\sum_{i=1}^k w^{(t)}_i\cdot \pr{\an{B_i,P_t} -\beta_i}}\)

    \If {\(y^{(t)} -g(X_t) \leq \epsilon \)}
        \State break
    \EndIf
    
    \State \(t\leftarrow t+1\)
    \EndWhile
    \State return \(X^{(t)}\) 
\end{algorithmic}
\end{algorithm}

\hide{\begin{algorithm} 
\caption{Multiplicative weight update (MW) for \gfpca{}} 
\label{alg:mw-old} 
    \SetKwInOut{Input}{Input}
    \SetKwInOut{Output}{Output}

    \Input{\(\alpha,\beta\in\R,\ A\in\R^{m_1\times n},B\in\R^{m_2\times n}\), \(\eta>0\), positive integer \(T\)}
    \Output{\(\argmin\limits_{P,z}\ \ z \text{, \ subject to}\) \\ \(z\geq\alpha - \frac{1}{m_1}  \langle A^\top A,P\rangle\),\\ \(z\geq\beta - \frac{1}{m_2}\langle B^\top B,P\rangle\), \\\(P\in\cP=\{M\in\R^{n\times n}:0\preceq M \preceq I,\text{Tr}(M)\leq d\}\)}

    Initialize \(p^0=(1/2,1/2)\)\;
    \For{\(t=1,\ldots,T\)}{
    \((P_t,m_1^t,m_2^t)\leftarrow \oracle(p^{t-1},\alpha,\beta,A,B)\)\;
    \(\hat{p}_i^t\leftarrow p_i^{t-1}e^{\eta m_i^t}\), for \(i=1,2\)\;
    \(p_i^t\leftarrow \hat{p}_i^t/(\hat{p}_1^t+\hat{p}_2^t)\), for \(i=1,2\)\;
    }
    \Return \(P^*=\frac{1}{T}\sum_{t=1}^T P_t\) , \(z^*=\max\{\alpha - \frac{1}{m_1}  \langle A^\top A,P^*\rangle,\beta- \frac{1}{m_2}\langle B^\top B,P^*\rangle\}\)
\end{algorithm}}

\paragraph{Runtime analysis.} The convergence of MW is stated as follows and can be derived from the convergence of mirror descent \cite{nemirovsky1983problem}. A proof can be found in Appendix \ref{sec:proof-convergence}.

\begin{theorem} \label{thm:MW-convergence} Given PSD \(B_1,\ldots,B_k\in\R^{n\times n}\) and \(\beta_1,\ldots,\beta_k\in\R\), we consider a maximization problem of  the function \[g(X)=\min_{i\in[k]}  \an{B_i,X}-\beta_i\] over \(\sX=\set{X\in\R^{n \times n}:\tr(X)=d,0\preceq X\preceq I}\).
For any \(T\geq1\), the \(T\)-th iterate \(X_T\) of multiplicative weight update algorithm with uniform initial weight and learning rate \(\eta=\sqrt{\frac{\log k }{2T}}L\) where \(L=\max_{i\in[k]} \tr(B_i)\) satisfies
\begin{equation}
g^*-g(X_T) \leq \sqrt{\frac{2\log k}{T}} \max_{i\in[k]} \tr(B_i)
\end{equation} 
where \(g^*\) is the  optimum of the maximization problem.  
\end{theorem}

Theorem \ref{thm:MW-convergence} implies that MW takes \(O\pr{\frac{\log k}{\epsilon^2}}\) iterations to obtain an additive  error bound \(\epsilon\). For the \gfpca{} application, \(\max_{i\in[k]} \tr(B_i)\) is a scaling of the data input and can be bounded if data are normalized. In particular, suppose \(B_i=A_i^TA_i\) and each column of \( A_i\) has mean zero and variance at most one, then 
\[\max_{i\in[k]} \tr(B_i) =\max_{i\in[k]}\norm{A_i}^2_F \leq n.  \]

\paragraph{MW for two groups.}
For  MM-Var and MM-Loss objectives in two groups, the simplex is a one-dimensional segment. The dual problem \eqref{eq:sdp-dual-raw} reduces to
\begin{align}
\inf_{w\in[0,1]} \pr{h(w):=\max_{\substack{X\in\R^{n \times n}\\ \tr(X)=d,0\preceq X\preceq I}} \an{wB_1+(1-w)B_2,X} } \label{eq:dual-2-groups}
\end{align}
The function \(h(w)\)
is a maximum of  linear functions \(\an{wB_1+(1-w)B_2,X}\) in \(w\), and hence is convex on \(w\). 
Instead of mirror descent, one can apply ternary search, a technique applicable to maximizing a general  convex function in one dimension, to solve \eqref{eq:dual-2-groups}. However, we claim that  binary search, which is faster than ternary search, is also a valid choice.

First, because \(h(w)\) is convex, we may assume that \(h\) achieves minimum at \(w=w^*\) and that all subgradients \(\partial h(w)\subseteq(-\infty,0]\) for all \(w<w^*\) and \(\partial h(w)\subseteq [0,\infty)\) for all \(w>w^*\).
In the binary search algorithm with current iterate \(w=w_t\), let 
\begin{align*}
X_t \in \argmax_{\substack{X\in\R^{n \times n}\\ \tr(X)=d,0\preceq X\preceq I}} \an{w_tB_1+(1-w_t)B_2,X} 
\end{align*} 
be any solution of the optimization (which can be implemented easily by  standard PCA). Because a linear function \(\an{wB_1+(1-w)B_2,X_t}=\an{B_2,X_t}+w\an{B_1-B_2,X_t}\) is a lower bound of \(h(w)\) over \(w\in[0,1]\) and \(h\) is convex, we have \(\an{B_1-B_2,X_t}\in\partial h(w_t)\). 
Therefore, the binary search algorithm can check the sign of \(\an{B_1-B_2,X_t}\) for a correct recursion. If \(\an{B_1-B_2,X_t}<0\), then \(w^*>w_t\); if \(\an{B_1-B_2,X_t}>0\), then \(w^*<w_t\); and the algorithm recurses in the left half or right half of the current segment accordingly. If \(\an{B_1-B_2,X_t}=0\), then \(w_t\) is an optimum dual solution.
\paragraph{Tuning in practice.}
In practice for MM-Var and MM-Loss objectives, we tune the learning rate of mirror descent much higher than in theory. In fact, we  find that the last iterate of MW sometimes converges (certified by checking the duality gap), and in such case the convergence is much faster. For NSW, the dual is still a convex optimization, so standard technique such as gradient descent can be used. We found that in practice, however, the unboundedness of the feasible set is a challenge to tune MW for NSW to converge quickly.

\subsubsection{Frank-Wolfe (FW)}
Observe that while the original optimization \eqref{eq:sdp-gen-obj-in-z}-\eqref{eq:sdp-gen-con-in-z}, which is in the form
\[ \max_{\substack{X\in\R^{n \times n}\\ \tr(X)=d,0\preceq X\preceq I}} g(z(X))\]
where the utility \(z\) is a function of \(X\), is a nontrivial convex optimization, its linear counterpart
\begin{align*}
\max_{\substack{X\in\R^{n \times n}\\ \tr(X)=d,0\preceq X\preceq I}} \an{C,X}
\end{align*}
is easily solvable by standard PCA for any given matrix \(C\). This motivates Frank-Wolfe (FW)\ algorithm \cite{frank1956algorithm} which requires a linear oracle (solving the same problem but with a linear objective) in each step.
The instantiation of FW to \gfpca{} is summarized in Algorithm \ref{alg:fairPCA-FW}. 
\begin{algorithm}[h]
\caption{Frank-Wolfe Algorithm for Multi-Criteria Dimensionality Reduction}\label{alg:fairPCA-FW}
\begin{algorithmic}[1]
\Statex \textbf{Input:} \(B_1,\ldots,B_k\in\R^{n\times n}\), \(d\leq n\), concave \(g:\R^k\rightarrow \R\), learning rate \(\eta_t\), duality gap target
\(\epsilon\)\Statex \textbf{Output:} A matrix \(X\in\R^{n\times n}\)  that maximizes \(g(\an{B_1,X},\ldots,\an{B_k,X})\) subject to \(\tr(X)=d,0\preceq X\preceq I\) with additive error at most \(\epsilon\) from the optimum
\State Initialize a feasible \(X_0\) (we use \(X_0=\frac dn I_n\)), \(t=0\), and duality gap \(g_0=\infty\) 
\While {\(g_t>\epsilon\)}
\State \(G_t \leftarrow \nabla_X g(X_t)\)
\State \(S_t \leftarrow VV^T\) where \(V\) is  \(n\)-by-\(d\) matrix of top \(d\) eigenvectors of \(G_t\) 
\Comment Linear oracle of FW 
\State  \(X_{t+1}\leftarrow (1-\eta_t)x_t + \eta_t S_t\)
\State \(g_t\leftarrow (S_t - X_t ) \cdot G_t \)
\Comment Duality gap
\State \(t\leftarrow t+1\)
\EndWhile
\State Output $X_t$
\end{algorithmic}
\end{algorithm}

One additional concern for implementing FW is obtaining gradient \(\nabla_X g(X_t)\). For some objectives such as NSW, this gradient can be calculated analytically and efficiently (some small error may need to be added to stabilize the algorithm from exploding gradient when the variance is close to zero; see \(\lambda\)-smoothing below). Other objectives, such as MM-Var and MM-Loss, on the other hand, are not differentiable everywhere. Though one may try to still use FW and calculate gradients at the present point (which is differentiable with probability one), there is no theoretical guarantee  for the FW convergence when the function is non-differentiable (even when the feasible set is compact as in our SDP relaxation). Indeed, we find that FW does not converge in our experiment settings. 

There are modifications of FW which has convergence guarantee for maximizing concave non-differentiable functions. The algorithm by \citet{white1993extension} (also used by \citet{ravi2019deterministic}) requires a modified linear oracle, namely a maximization of \(\min_{\delta\in T(X_t,\epsilon)}\delta\cdot(Y-X_t)\) over  \(Y\) in SDP feasible set where \( T(X_t,\epsilon)\) is the set of all subgradients at all points in the \(\epsilon\)-neighborhood of \(X_t\). In our setting, if the neighborhood has at least two distinct subgradients, then the oracle reduces to  the  form at least as hard as the original problem, making the method unhelpful. Another method is by smoothing the function, and the natural choice is by replacing \(g(z)=\min_{i\in[k]}z_i\) by \(g_\mu(z)=\mu\log\pr{\sum_{i\in[k]}e^{\frac{z_i}{\mu}}}\) where \(\mu>0\) is the smoothing parameter. We find that the error from approximating \(g\) by \(g_\mu\) is significant for even moderate \(\mu\geq10^{-2}\), and any smaller \(\mu\) causes substantial numerical errors in computing the exponent. Hence, we do not find FW nor its modification useful in practice for  solving MM-Var or MM-Loss objective.

\paragraph{Runtime analysis.} The NSW objective can be ill-conditioned or even undefined when matrices \(B_i\)'s span low dimensions and those dimensions are distinct.
This is due to \(\log\an{B_i,X}\) not defined at \(\an{B_i,X}=0\) and having unbounded gradient when \(\an{B_i,X}\) is small. To stabilize the objective, we solve the   \(\lambda\)-smoothed NSW objective
\begin{equation}
\max_{\substack{X\in\R^{n \times n}\\ \tr(X)=d,0\preceq X\preceq I}} g(X):=\sum_{i\in[k]} \log \pr{\an{B_i,X}+\lambda \cdot\norm{B_i}_F}
\label{eq:smooth-NSW}
\end{equation}  
Here,  \(\lambda>0\) is a regularizer, and the regularization term \(\lambda \cdot\norm{B_i}_F\) normalizes the effect of each group when their variances are small, regardless of their original sizes \(\norm{B_i}_F\). Since \(g(X)\) is now also Lipschitz, FW has the convergence guarantee as follows, which follows from the standard FW convergence bound \cite{jaggi2013revisiting}.  The proof can be found in Appendix \ref{sec:NSW-convergence}. \begin{theorem} \label{thm:FW-reg-convergence}
Given \(n,k,d,\lambda\) and PSD \(B_1,\ldots,B_k\in\R^{n\times n}\) as an input to \(\lambda\)-smooth NSW \eqref{eq:smooth-NSW}, the \(t\)-th iterate \(X_t\) of Frank-Wolfe with step sizes \(\eta_s=2/(s+1)\) 
satisfies 
\begin{equation}
g^*-g(X_t) \leq \frac{8kd}{\lambda(t+2)}
\end{equation} 
where \(g^*\) is the  optimum of the optimization problem.  
\end{theorem}

Theorem \ref{thm:FW-reg-convergence} implies that FW takes \(O\pr{\frac{kd }{\lambda\epsilon}}\) iterations to get an error bound \(\epsilon\) for solving \(\lambda\)-smooth NSW.

\paragraph{Tuning in practice.}
In practice, we experiment with more aggressive learning rate schedule and line search algorithm. We found that FW converges quickly for NSW objective. However, FW does not converge for MM-Var and MM-Loss for any learning rate schedule, including the standard \(\eta_t=\frac{1}{t+2}\) and line search. This is  consistent with, as we mentioned, that FW does not have convergence guarantee for non-differentiable objectives.

\subsubsection{Empirical runtime result on a large data set}

We  perform MW and FW heuristics on a large 1940 Colorado Census data set \citep{1940census}. The  data is preprocessed by one-hot encoding all discrete columns, ignoring columns with N/A, and normalizing the data to have mean zero and variance one on each feature. The preprocessed data set contains 661k data points and 7284 columns. Data are partitioned into 16 groups based on 2 genders and 8 education levels. We solve the SDP relaxation of \gfpca{} with MM-Var, MM-Loss, and NSW objectives until achieving the duality gap of no more than 0.1\% (in the case of NSW, the product of variances, not the sum of logarithmic of variances, is used to calculate this gap). The runtime results, in seconds, are in shown in Table \ref{tab:runtime-MW-FW-large-data}.
When \(n\) increases, the bottleneck of the experiment becomes the standard PCA itself. Since speeding up the standard PCA is not in the scope of this work, we capped the original dimension of data by selecting the first \(n\) dimensions out of 7284, so that the standard PCA can still be performed in a reasonable amount of time.
We note that the rank violation of solutions are almost always zero, and are exactly one when it is not zero, similarly to what we found for 
Adult Income and Credit data sets \citep{adult-data,default-dataset}.
\begin{table}
\caption{Runtime of MW and FW for solving \gfpca{} on different fairness objectives and numbers of dimensions on the original  1940 Colorado Census data set. Runtimes of the standard PCA by SVD are included for comparison. All instances are solved to duality gap of at most 0.1\%. Times  are in second(s).}
\label{tab:runtime-MW-FW-large-data}
\centering
\begin{tabular}{|c|c|c|c|c|}\hline
Original Dimensions & MM-Var (by MW) & MM-Loss (by MW) & NSW (by FW) & Standard PCA (by SVD)\\\hline
\(n=1000\) & 77 & 65 & 11 & 0.22\\\hline
\(n=2000\) & 585 & 589 & 69 & 1.5\\\hline
\end{tabular}

\end{table}

\paragraph{Runtime of MW.}
We found that MM-Var and MM-Loss objectives are  solved  efficiently by MW,  whereas MW with gradient descent on the dual of NSW does not converge quickly. It is usual that the solution of the relaxation has rank exactly \(d\), and in all those cases we are able to tune learning rates so that the last iterate converges, giving a much faster convergence than the average iterate. For the Census data set, after parameter tuning, MW runs 100-200 iterations on both objectives.  MW for both Credit and Income data sets (\(n=23,59\))  on 4-6 groups on both objectives finishes in 10-20 iterations whenever the last iterate converges, giving a total runtime of  less than few seconds. Each iteration of MW takes  1x-2x of an SVD algorithm. Therefore, the price of fairness in PCA for MW-Var and MM-Loss objectives is 200-400x runtime for large data sets, and 20-40x runtime for medium data sets, as compared to the standard PCA without a fairness constraint. 

\paragraph{Runtime of FW.}
FW converges  quickly for NSW objective, and does not converge on MM-Var or MM-Loss. FW terminates in 10-20 iterations for Census Data. In practice, each iteration of FW has an overhead of 1.5x-3x of an SVD algorithm. We suspect codes can be optimized so that the constant overhead of each iteration is closer to 1x, as the bottleneck in each iteration  is one standard PCA. Therefore, the price of fairness in PCA for NSW objective is 15-60x runtime compared to the standard PCA without a fairness constraint.  


\bibliographystyle{plainnat}
\bibliography{reference}


\newpage
\appendix


\section{Proofs}
\label{app:proof}

\subsection{Proof of Theorem~\ref{thm:sameLoss} by local optimality }
Before proving Theorem~\ref{thm:sameLoss}, we first state some notations and lemmas.
We denote by \(A_1\in\R^{m_1\times n},A_2\in\R^{m_2\times n}\) data matrices of two groups, with rows as data points, in \(n\) dimensions that are to be projected onto \(d\) dimensions. Let \(\S_d\) denote the set of all \(d\)-dimensional subspaces of \(\R^n\). For $U\in\S_d$,
we let  \(P_U\in\PP_d\) be a matrix which has an orthonormal
basis of the subspace $U$ as its columns.

For each matrix \(M\), we let \(P^*_M\in \argmax_{Q\in \PP_d} \|MQ\|^2_F \). We denote \(loss(M,U)=\|MP^*_M\|^2_F - \|MP_U\|^2_F\) be the loss of data \(M\) by a projection onto \(U\). For data matrix \(A\) in \(n\) dimensions, we let \(g_A(U)=\|AP_U\|^2_F\) for \(U\in\S_d\). We let \(h(U)\) be the marginal loss objective \eqref{eq:marginal-loss} given by a projection matrix \(P_U\) to be minimized.
We note that \(g_A\), \(h\), and \(loss\) are well-defined since their values do not change based on the orthonormal basis representation of \(U\).

For a function \(f:\PP_d\rightarrow \R\), we call \(P^*\in\PP_d\) a local maximum of \(f\) if there exists \(\epsilon>0\) such that \(f( P^*)\geq f(P)\) for all \(P\in\PP_d\) with \(\norm{P- P^*}_F<\epsilon\) (we note that any \(P\in\PP_d\) always has an \(\epsilon\)-neighbor by a slight change in one column while remaining in the orthogonal complement of the rest for all \(d<n\)).  
We define an \(\epsilon\)-neighborhood  of \(U\in\S_d\) as  the set of all $d$-dimensional subspaces \(V\in\S_d\) such that there exist orthonormal bases \(P_U,P_V\) of \(U,V\) with \(\norm{P_U-P_V}_F < \epsilon\). (More generally, we let \(d(U,V):=\inf\set{\norm{P_U-P_V}_F :P_U,P_V \text{ are orthonormal bases of } U,V}\) be the metric on \(\S_d\).) We then define the local optimality and continuity of \(g_A\), \(h\), and \(loss\) accordingly. 

We state the a property of \(g_A\) which is a building block for the proof of Theorem~\ref{thm:sameLoss}. 
\begin{lemma}
\label{lem:C}
Given a matrix $A\in \mathbb{R}^{a\times n}$, the value of the function~$g_A$ at any local maximum is the same. 

\end{lemma}
Before we prove this lemma, we prove use it to prove Theorem~\ref{thm:sameLoss}.

\begin{proofof}{Theorem~\ref{thm:sameLoss}}
We prove by a contradiction.
Let $W$ be a global minimum of $h$ and assume that
\begin{equation}
\label{eq:bigger}
loss(A_1,W) > loss(A_2,W). 
\end{equation}Since $loss$ is continuous, there exists \(\epsilon>0\) such that for any \(W_\epsilon\) in the \(\epsilon\)-neighborhood  of $W$, $h(W_\epsilon)= loss(A_1,W) $. Since $W$ is a global minimum of $h$, it is a local minimum of $loss(A_1,W),  $ or equivalently a local maximum of 
$g_{A_1}$.

Let $\{v_1, \ldots, v_n\}$ be an orthonormal basis of the eigenvectors
of $A_1^TA_1$ corresponding to eigenvalues
$\lambda_1 \geq \lambda_2 \geq \ldots \geq \lambda_n$. Let $V^*$ be
the subspace spanned by $\{v_1,\ldots,v_d\}$. Note that
$loss(A_1,V^*)=0$. Since the loss is always non-negative for
both $A_1$ and $A_2$, \eqref{eq:bigger} implies that $loss(A_1, W) > 0$.
Therefore, $W\neq V^*$ and $g_{A_1}(V^*) > g_{A_1}(W)$. By Lemma~\ref{lem:C}, this is in contradiction with $V^*$ being a global maximum and $W$ being a local maximum of $g_{A_1}$.
\end{proofof} 

\hide{
\begin{lemma}
\label{lem:B}
        Given a matrix $V=[v_1,\ldots, v_d]\in \mathbb{R}^{n\times d}$ with orthonormal columns, we have:
        \begin{itemize}
                \item[$\diamond$] $loss(A, V) = \|AP_A^*\|_F^2 - \sum_{i=1}^{d} \|Av_i\|^2 = \|AP_A^*\|_F^2 - \langle A^TA, VV^T\rangle$
                \item[$\diamond$] $\|A-AVV^T\|_F^2 =  \|A\|_F^2 - \|AV\|_F^2 = \|A\|_F^2 - \sum_{i=1}^d \|Av_i\|^2$
        \end{itemize}
\end{lemma}}

We now prove Lemma~\ref{lem:C}. We will also use the following formula: given a matrix $V=[v_1,\ldots, v_d]\in \mathbb{R}^{n\times d}$ with orthonormal columns and \(U=\sp(V)\), we have       \begin{equation}
g_A(U)=\norm{AV}_F^2=\sum_{i=1}^d \|Av_i\|^2 \label{eq:to-vec}
\end{equation}
Formula \eqref{eq:to-vec} is straightforward from the definitions of \(g_A\) and Frobenius norm.

\begin{proofof}{Lemma~\ref{lem:C}}
We prove that the value of function $g_A$ at its local maxima is equal to its value at its global maximum, which we know is the subspace spanned by a top $d$ eigenvectors of $A^TA$. 
Let $\{v_1, \ldots, v_n\}$ be an orthonormal basis of eigenvectors of $A^TA$ with corresponding eigenvalues
$\lambda_1 \geq \ldots \geq \lambda_n$ where ties are
broken arbitrarily. Let $V^*$ be the subspace spanned by
$\{v_1,\ldots, v_d\}$ and let \(U\in\S_d\) be a local optimum of \(g_A\). We assume for contradiction that $g_A(U) < g_A(V^*)$. We will show that there exists a small constant \(C>0\) such that for all \(\epsilon\in(0,C)\), there exists \(U_\epsilon\) in the \(\epsilon\)-neighborhood of \(U\) which strictly increases \(g_A\). This will contradict the local optimality of \(U\). 

As $g_A(U) < g_A(V^*)$, we have $U\neq V^*$. 
Let $k$ be the smallest index such that $v_k \notin U$. Extend $\{v_1,\ldots, v_{k-1}\}$ to an orthonormal basis of $U$: $\{v_1, \ldots, v_{k-1}, v'_k, \ldots, v'_d\}$. Let $q \geq k$ be the smallest index such that $\|Av_q\|^2 > \|Av'_q\|^2$ (such an index $q$ must exist because $g_A(U)< g_A(V^*)$). Without loss of generality, we can assume that $q=1$. Therefore, we assume that $v_1$, the top eigenvector of $A^TA$, is not in $U$ and that \(v_1\) strictly maximizes the function $\|Au\|^2$ over the space of unit vectors $u$. Specifically, for any unit vector $u\in U$, $\|Au\|^2<  \|Av_1\|^2 = \lambda_1$. We distinguish two cases:

\paragraph{Case $v_1\perp U$.}  Let $w_\epsilon = \sqrt{1-\epsilon^2} u_1 + \epsilon v_1$. Then we have $\|w_\epsilon\|=1$ and that $\{w_\epsilon, u_2, \ldots, u_d\}$ is an orthonormal set of vectors. We set $U_{\epsilon}= \sp\{w_\epsilon, u_2, \ldots, u_d\} $ and  claim that $g_A(U_\epsilon) - g_A (U)>0$ for all \(\epsilon\in(0,1)\). By  \eqref{eq:to-vec}, $g_A(U_\epsilon)=\|Aw_\epsilon\|^2 + \|Au_2\|^2+\ldots+\|Au_d\|^2 \) and \( g_A (U)= \|Au_1\|^2+\|Au_2\|^2+\ldots + \|Au_d\|^2$. Hence,   $g_A(U_\epsilon) - g_A (U)=\|Aw_\epsilon\|^2 - \|Au_1\|^2$.
We have
\begin{align*}
\|Aw_\epsilon\|^2 - \|Au_1\|^2 &= \|A(\sqrt{1-\epsilon^2}u_1+\epsilon v_1)\|^2 - \|Au_1\|^2\\
&=(\sqrt{1-\epsilon^2}u_1^T+\epsilon v_1^T)A^TA(\sqrt{1-\epsilon^2}u_1+\epsilon v_1) - \|Au_1\|^2\\
&=(1-\epsilon^2)u_1^TA^TAu_1+\epsilon^2v_1^TA^TAv_1+2\sqrt{1-\epsilon^2}\epsilon u_1^TA^TAv_1 - \|Au_1\|^2\\
&=(1-\epsilon^2)\|Au_1\|^2+\epsilon^2\lambda_1+2\epsilon \sqrt{1-\epsilon^2} u_1^TA^TAv_1 - \|Au_1\|^2 \\
&= \epsilon^2 (\lambda_1 - \|Au_1\|^2) + 2\epsilon \sqrt{1-\epsilon^2} u_1^TA^TAv_1 
\end{align*}
Next, we have $ u_1^TA^TAv_1 = u_1^T (\lambda_1 v_1)= \lambda_1 u_1^Tv_1 =0$ since $v_1$ is an eigenvector of $A^TA$ and $v_1 \perp u_1$. This and  the fact that $\|Au_1\|^2 < \lambda_1$ give
\begin{align*}
        \|Aw_\epsilon\|^2 - \|Au_1\|^2 &= \epsilon^2(\lambda_1 - \|Au_1\|^2) > 0
\end{align*}
as desired.
\paragraph{Case  $v_1\not\perp U$.} Let $v_1 = \sqrt{1-a^2}z_1 + az_2$ where $z_1 \in U$, $z_2 \perp U$, \(a\in(0,1) \) and$\|z_1\|=\|z_2\|=1$, so the projection of $v_1$ to $U$ is $\sqrt{1-a^2}z_1$. Note that $z_2 \neq 0$  and \(a>0\) since  $v_1 \notin U$, and that \(a<1\) since \(v_1\not\perp U\). We extend $\{z_1\}$ to an orthonormal basis of $U$:  $\{z_1,u_2, \ldots, u_k\}$. 

Consider the unit vector $w_\epsilon   = \sqrt{1-\epsilon^2}z_1 + \epsilon z_2$ for \(\epsilon\in(0,1)\). Let $U_\epsilon:=\sp \{w_\epsilon,u_2,\ldots,u_d\}$. Note that \( \{w_\epsilon,u_2,\ldots,u_d\}\) is orthonormal for any \(\epsilon\) since both $z_1$ and $z_2$ are orthogonal to all of $u_2, \ldots, u_d$  and $w_\epsilon$ is in the span of $z_1, z_2$. 
Since the chosen orthonormal bases of \(U_\epsilon\) and \(U\)  differ only in $w_\epsilon$ and $z_1$, by  \eqref{eq:to-vec}, \(g_A(U_\epsilon) - g_A(U)=\|Aw_\epsilon\|^2 - \|Az_1\|^2$.
We can write
\begin{align*}
w_\epsilon &= \left(\sqrt{1-\epsilon^2} - \frac{\epsilon\sqrt{1-a^2}}{a}\right)z_1 +  \frac{\epsilon}{a}\left(\sqrt{1-a^2}z_1 + az_2\right)\\
&=\left(\sqrt{1-\epsilon^2} - \frac{\epsilon\sqrt{1-a^2}}{a}\right)z_1 +  \frac{\epsilon}{a}v_1.
\end{align*}
Thus, by $A^TAv_1 = \lambda_1 v_1$ 
(as $v_1$ is an eigenvector with eigenvalue $\lambda_1$) and $z_1^Tv_1=\sqrt{1-a^2}$,
we have\begin{align*}
\|Aw_\epsilon\|^2 
&= \left( \sqrt{1-\epsilon^2} - \frac{\epsilon\sqrt{1-a^2}}{a} \right)^2\|Az_1\|^2 + \frac{\epsilon^2}{a^2} \|Av_1\|^2 + 2\frac{\epsilon}{a}\left(\sqrt{1-\epsilon^2} - \frac{\epsilon\sqrt{1-a^2}}{a}\right)z_1^TA^TAv_1\\
&=\left(1-\epsilon^2 + \frac{\epsilon^2(1-a^2)}{a^2}-2\frac{\epsilon\sqrt{(1-\epsilon^2)(1-a^2)}}{a}\right)\|Az_1\|^2 + \frac{\epsilon^2}{a^2}\lambda_1 \\
& + 2\frac{\epsilon}{a}\left(\sqrt{1-\epsilon^2} - \frac{\epsilon\sqrt{1-a^2}}{a}\right)\lambda_1 z_1^Tv_1\\
&=\left(1-2\epsilon^2 + \frac{\epsilon^2}{a^2}-2\frac{\epsilon\sqrt{(1-\epsilon^2)(1-a^2)}}{a}\right)\|Az_1\|^2 \\
&+ \left(\frac{\epsilon^2}{a^2} + 2\frac{\epsilon\sqrt{(1-\epsilon^2)(1-a^2)}}{a} - 2\frac{\epsilon^2(1-a^2)}{a^2}\right)\lambda_1 \\
&=\|Az_1\|^2 + (\lambda_1-\|Az_1\|^2)\left(2\frac{\epsilon\sqrt{(1-\epsilon^2)(1-a^2)}}{a}+2\epsilon^2 - \frac{\epsilon^2}{a^2}\right) >\|Az_1\|^2.
\end{align*}
The last inequality follows since $\lambda_1 > \|Az_1\|^2$. Now, we let $ 0 < \epsilon < \frac{1}{1+b}$ for $b= 4a^2(1-a^2)$ so that $ 2\frac{\epsilon\sqrt{(1-\epsilon^2)(1-a^2)}}{a} > \frac{\epsilon^2}{a^2}$ .
Then, $\|Aw_\epsilon\|^2 > \|Az_1\|^2$ and therefore $g_A(U_\epsilon) > g_A(U)$ for all such \(\epsilon\).
\end{proofof}

\subsection{Proof  of Theorem~\ref{thm:sameLoss} by SDP relaxation}
We remark that Theorem \ref{thm:sameLoss} also follows from using SDP relaxation formulation of \gfpca{} for marginal loss objective, which is
\begin{align}
  \qquad \min_{X\in\R^{n\times n}} &\, \, z \ \text{ subject to } \label{eq:sdp-mar-loss-2-groups-top}\\
\beta_i-\langle A_i^TA_i, X\rangle & \leq z \quad ,i \in \{1,2\}\\
\tr(X) & \le d\\
0 \preceq \, X  &\preceq I\label{eq:sdp-mar-loss-2-groups-bottom}
\end{align}
where \(\beta_i=\max_{Q\in\PP_d} \|A_i Q\|^2 \). We provide a proof here as another application of the relaxation. 

\begin{proofof}{Theorem~\ref{thm:sameLoss}}
Let \(X^*\) be an extreme solution of the SDP \eqref{eq:sdp-mar-loss-2-groups-top}-\eqref{eq:sdp-mar-loss-2-groups-bottom}. Suppose the marginal loss of two groups are not equal; without loss of generality, we have
\[ \beta_1-\langle A_1^TA_1, X\rangle >\ \beta_2-\langle A_2^TA_2, X\rangle.\] 
Since the constraint \(\beta_2-\langle A_2^TA_2, X\rangle\leq z\) is not tight, we can delete  it and the new  SDP does not change the optimal solution. However, an optimal solution of the new SDP, which now has only one group, is a standard PCA solution of the first group. This solution gives a loss of zero, so the optimum of the new SDP is zero. Therefore, the optimum of original SDP \eqref{eq:sdp-mar-loss-2-groups-top}-\eqref{eq:sdp-mar-loss-2-groups-bottom} is also zero. However, since the losses of both groups are always non-negative, they must be zero and hence are  equal.
\end{proofof}

\subsection{Proof of MW convergence} \label{sec:proof-convergence}
Here we prove Theorem  \ref{thm:MW-convergence} using the mirror descent convergence \cite{nemirovsky1983problem}. When mirror descent is performed over a simplex, and the convergence guarantee from mirror descent is simplified as follows. We write \(\Delta_k:=\set{w\in\R^k:w\geq 0, \sum_{i\in[k]}w_i=1}\).
\begin{theorem}
 \label{thm:simplex-convergence} (\cite{nemirovsky1983problem})
Consider a problem of maximizing  a concave  function \(h(w)\) over \(w\in\Delta_k\) where \(h(u)-h(v)\leq L\norm{u-v}_1\) for \(u,v\in\Delta_k\).
 The \(t\)-th iterate \(w^{(t)}\) of the mirror descent algorithm with negative entropy potential function \(R( w)=\sum_{i=1}^k w_i\log w_i\), step size \(\eta\), and the initial solution \(w^{(0)}=(1/k,\ldots,1/k)\) satisfies
\begin{equation}
h^*-h(w^{(t)}) \leq\frac{\log k}{\eta t} + \frac{\eta}{2}L^2
\end{equation} 
where \(h^*\) is the  optimum of the maximization problem. 
\end{theorem}
We apply Theorem \ref{thm:simplex-convergence} to obtain the convergence bound for \gfpca{} for fairness criteria in the  "max-min" type including MM-Var and MM-Loss. 
\begin{proof}[Proof of Theorem \ref{thm:MW-convergence}]
We showed earlier that the dual problem of maximizing \(g\) is \begin{align}
\inf_{w\in\Delta_k} \pr{ h(w):=\max_{\substack{X\in\sX}}\sum_{i=1}^k w_i\an{B_i,X}} \nonumber
\end{align}
and that MW algorithm is equivalent to mirror descent on \(h(w)\) over \(w\in\Delta_k\). Applying Theorem \ref{thm:simplex-convergence} and substituting \(\eta=\sqrt{\frac{\log k }{2T}}L\), the desired convergence bound follows, so it remains to show that \(h(u)-h(v)\leq L\norm{u-v}_1\) for \(u,v\in\Delta_k\).
Let \(X_u^*\in\argmax_{{X\in \sX}}\sum_{i=1}^k u_i\an{B_i,X}\). We have
\begin{align*}
h(u)=\sum_{i=1}^k u_i\an{B_i,X_u^*}&=\sum_{i=1}^k v_i\an{B_i,X_u^*}+\sum_{i=1}^k (u_i-v_i)\cdot\an{B_i,X_u^*}\\
&\leq h(v)+\sum_{i=1}^k \left|u_i-v_i\right|\cdot\an{B_i,X_u^*}\\
&\leq h(v)+\sum_{i=1}^k \left|u_i-v_i\right|\cdot\tr(B_i)\leq h(v)+\norm{u-v}_1\cdot L
\end{align*} 
where the first inequality follows from \(B_i,X_u^*\succeq 0\) so that their inner product is non-negative, and the second follows from \(B_i\succeq 0\) and \(X^*_u\preceq I\).
This finishes the proof.
\end{proof}

\subsection{Proof of NSW convergence} \label{sec:NSW-convergence}
We have a standard convergence  of FW for differentiable and \(L\)-Lipschitz objective functions as follows.   

\begin{theorem} \label{thm:FW-convergence} (\cite{jaggi2013revisiting})
Consider a maximization problem of an \(L\)-Lipschitz concave function \(g(X)\) over a convex feasible set \(D\) of diameter \(\diam(D):=\max_{X,Y\in D}\norm{X-Y}^2\).
The \(t\)-th iterate \(X_t\) of Frank-Wolfe with step sizes \(\eta_s=2/(s+1)\) satisfies
\begin{equation}
g^*-g(X_t) \leq \frac{2L\cdot \diam^2(D)}{t+2}
\end{equation} 
where \(g^*\) is the  optimum of the maximization problem.  
\end{theorem}
The theorem gives the proof of Theorem \ref{thm:FW-reg-convergence} as follow(s).
\begin{proof}[Proof of Theorem \ref{thm:FW-reg-convergence}]
We first show that \(g\) is \(L\)-Lipschitz for \(L=\frac k\lambda\). We have
\begin{displaymath}
\norm{\nabla_X g(X)}_F=\norm{\sum_{i\in[k]} \frac{B_i}{\an{B_i,X}+\lambda \cdot\norm{B_i}_F}} \leq\sum_{i\in[k]}{\frac{\norm{B_i}_F}{\lambda \cdot\norm{B_i}_F}} =\frac k\lambda
\end{displaymath}
as claimed. We write \(\sX=\set{X\in\R^{n \times n}:\tr(X)=d,0\preceq X\preceq I}\). By the convergence of FW in Theorem \ref{thm:FW-convergence}, it remains to show that \(\diam(\sX)\leq 2\sqrt d\). For any \(X\in\sX\), let \(\sigma_i(X)\) be the eigenvalues of \(X\) in the descending order. Then, we have \(\norm{X}_F^2=\sum_{i\in[n]}\sigma_i^2(X)\), \(\sum_{i\in[n]}\sigma_i(X)\leq d\), and \(\sigma_i(X)\in[0,1]\) for all \(i\) by the constraints in \(\sX\). Since a function \(h(x)=x^2\) is convex, \(\sum_{i\in[n]}\sigma_i^2(X)\) is maximized  subject to \(\sum_{i\in[n]}\sigma_i(X)=1\) when \(\sigma_1(X)=\ldots=\sigma_d(X)=1\) and \(\sigma_{d+1}(X)=\ldots=\sigma_n(X)=0\). So, we have \(\norm{X}_F^2\leq d\), and therefore
\begin{equation}
\max_{X,Y\in\sX} \norm{X-Y}_F \leq \max_{X\in\sX} \norm{X}_F +\ \max_{Y\in\sX} \norm{Y}_F \leq 2\sqrt d
\end{equation}
as needed.
\end{proof}

\section{Tightness of the rank violation bound} \label{sec:tight}
Here,  we  show that the bound of rank of extreme solutions in Theorem \ref{thm:low-rank} is tight by the following statement.

\begin{lemma} \label{lem:low-rank-tight}
For  any \(n,d,\) and \(s\) such that \(d\leq n\) and \(1\leq s\leq n-d\), there exist \(m=\frac{(s+1)(s+2)}{2}-1\) real matrices $A_1,\ldots, A_m$ and a real matrix \(C\), all of dimension \(n\times n\), and \(b_1,\ldots b_m\in\RR\) such that \(\sdp\) is feasible and all of its solutions have rank at least \(d+s\).
\end{lemma}
The example of instance is modified from \cite{bohnenblust1948joint} (and can also be found in \cite{ai2008low}) to match our SDP that has the additional \(X\preceq I\) constraint.
\begin{proof}
We construct the constraints of \(\sdp\)   so that the feasible set is 
\begin{equation}
\set{X=\begin{bmatrix}\lambda I_{s+1} & X_{12} \\
X_{21} & X_{22} \\
\end{bmatrix} \in\R^{n\times n}: \lambda\in\R,\ \tr(X)\leq d,\ 0\preceq X\preceq I_n \label{eq:matrix-example}
} 
\end{equation}
This can be done by \(m\)  constraint of the form \eqref{eq:sdp-con} (e.g. by using \(\frac{s(s+1)}{2}\) equality constraints to set off-diagonal entries of top \((s+1)\times(s+1)\) submatrix to zero, and the rest to set diagonal entries to be identical). To finish the construction of the instance, we set 
\[C=\diag(\underbrace{0,\ldots,0}_{s \text{ times}},\underbrace{1,\ldots,1}_{d \text{ times}},\underbrace{0,\ldots,0}_{n-s-d \text{ times}})\]
Then, we claim that \(\sdp\) has an extreme solution
\begin{equation}
X^*=\diag(\underbrace{1/s,\ldots,1/s}_{s+1 \text{ times}},\underbrace{1,\ldots,1}_{d-1 \text{ times}},\underbrace{0,\ldots,0}_{n-s-d \text{ times}}) \label{eq:X-tight-example}
\end{equation}
We first show that \(X^*\) is  an optimal solution. Let \(x_1,\ldots,x_n\) be diagonal entries of \(X\) in the feasible set. By \(0\preceq X\preceq I\), we have \(x_i= \mathbf{e_i}^TX\mathbf{e_i}\in[0,1]\) where \(\mathbf{e_i}\) is the unit vector at coordinate \(i\). This fact, combined with \(\tr(X)=\sum_{i=1}^n x_i\leq d\), shows that \(\an{C, X}\) is maximized when \(\diag(X)\) is as described in \eqref{eq:X-tight-example}.

We now show that \(X^*\) is extreme. In fact, we show that \(X^*\) is the unique solution and hence necessarily extreme (since the feasible set is convex and compact). From the argument above, any solution \(\bar X=[x_{ij}]_{i,j\in[n]}\) must satisfy \(\diag(\bar X) = \diag(X^*)\), and it remains to show that off-diagonal entries must be zero. Because every \(2\times2\) principle minor of \(\bar X\) is non-negative, we have \(x_{ij}=0\) if \(i\in\set{s+d+1,\ldots,n}\) or \(j\in\set{s+d+1,\ldots,n}   \). We know that \(x_{ij}=0\) for all \(i,j\leq s+1\) such that \(i\neq j\) by the constraints in \eqref{eq:matrix-example}, so it remains to show that
\(x_{ij}=0\) for \(s+2\leq i \leq d+s\) and \(j\leq d+s\) such that \(i\neq j\). 

Let \(i,j\) be one of such pair. We claim that the  bigger eigenvalue of the \(2\times\ 2\) submatrix 
\(\begin{bmatrix}
x_{ii} & x_{ij} \\
x_{ij} & x_{jj} \\
\end{bmatrix}\) 
is at least \(\max\set{x_{ii},x_{jj}}\), with a strict equality if and only if \(x_{ij}\neq 0\). Let \(p(t)\) be the characteristic polynomial of the matrix. Observe that \(p(x_{ii}),p(x_{jj})\leq0\) (\(<0\) if \(x_{ij}\neq0\) and \(=0\) if \(x_{ij}=0\)) and that \(p(t')>0\) for a sufficiently large \(t'>\max\set{x_{ii},x_{jj}}\). Hence, there must be a root of \(p(t)\) in between \(\max\set{x_{ii},x_{jj}}\) and \(t'\), and the claim follows. Now, for  the pair \(i,j\), we have \(x_{jj}=1\). By Cauchy's Interlacing Theorem, \(X\preceq I\) implies that any  \(2\times2\) submatrix of \(X\) has both eigvenvalues at most 1. Therefore, the claim implies that \(x_{ij}=0\).
\end{proof}

\end{document}